\pdfoutput=1 %
\documentclass[a4paper,
				twoside,
				]{article}

\usepackage[
			fullwidth, %
			nodate, %
			envcountsect, %
			runhead,
			noinfo,
			algopkg=algo2e, %
			bibpkg=biblatex, %
			authblk, %
			]{ducha}

\newcommand{\RR}[1]{{\color{red}#1}}
\newcommand{\RA}[1]{{\color{blue}#1}}
\renewcommand{\RR}[1]{{#1}}
\renewcommand{\RA}[1]{{#1}}

\usepackage[appendix=append, bibliography=common]{apxproof} %
\usepackage[capitalize]{cleveref}
\crefname{proposition}{Prop.}{Props.}
\crefname{theorem}{Thm.}{Thms.}
\crefname{corollary}{Cor.}{Cors.}
\crefname{figure}{Fig.}{Figs.}

\def\paperTitle{%
	On the Complexity of Distance-$d$ Independent Set Reconfiguration%
}
\def\paperAuthor{%
	D.A.~Hoang%
}
\titlerunning{%
	\paperTitle%
}%
\authorrunning{%
	\paperAuthor%
}
\relatedversion{A preliminary version of this paper~\cite{Hoang23} has appeared in the Proceedings of the 17th International Conference and Workshops on Algorithms and Computation (WALCOM 2023)}

\title{\paperTitle}
\author[1]{%
	Duc~A.~Hoang%
	\thanks{Present address: VNU University of Science, Vietnam National University, 334 Nguyen Trai Road, Thanh Xuan, Hanoi 100000, Viet Nam. Present email: \texttt{hoanganhduc@hus.edu.vn}.}%
	\addAuthInfo{c}{0000-0002-8635-8462}%
}
\affil[1]{%
	Graduate School of Informatics, Kyoto University\protect\\
	Kyoto, 
	Japan\protect\\
	\email{hoang.duc.8r@kyoto-u.ac.jp}
}

\begin{document}
\maketitle

\begin{abstract}
	For a fixed positive integer $d \geq 2$, a \textit{distance-$d$ independent set (D$d$IS)} of a graph is a \RR{vertex-subset} whose distance between any two members is at least $d$.
	Imagine that there is a token placed on each member of a D$d$IS.
	Two D$d$ISs are adjacent under Token Sliding ($\mathsf{TS}$) if one can be obtained from the other by moving a token from one vertex to one of its unoccupied adjacent vertices.
	Under Token Jumping ($\mathsf{TJ}$), the target vertex needs not to be adjacent to the original one.
	The \textsc{Distance-$d$ Independent Set Reconfiguration (D$d$ISR)} problem under $\mathsf{TS}/\mathsf{TJ}$ asks if there is a corresponding sequence of adjacent D$d$ISs that transforms one given D$d$IS into another.
	The problem for $d = 2$, also known as the \textsc{Independent Set Reconfiguration} problem, has been well-studied in the literature and its computational complexity on several graph classes has been known.
	In this paper, we study the computational complexity of \textsc{D$d$ISR} on different graphs under $\mathsf{TS}$ and $\mathsf{TJ}$ for any fixed $d \geq 3$.
	On chordal graphs, we show that \textsc{D$d$ISR} under $\mathsf{TJ}$ is in $\mathtt{P}$ when $d$ is even and $\mathtt{PSPACE}$-complete when $d$ is odd.  
	On split graphs, there is an interesting complexity dichotomy: \textsc{D$d$ISR} is $\mathtt{PSPACE}$-complete for $d = 2$ but in $\mathtt{P}$ for $d=3$ under $\mathsf{TS}$, while under $\mathsf{TJ}$ it is in $\mathtt{P}$ for $d = 2$ but $\mathtt{PSPACE}$-complete for $d = 3$.
	Additionally, certain well-known hardness results for $d = 2$ on perfect graphs and planar graphs of maximum degree three and bounded bandwidth can be extended for $d \geq 3$.
\end{abstract}

\section{Introduction}%
\label{sec:introduction}

Recently, \textit{reconfiguration problems} have attracted the attention from both theoretical and practical viewpoints.
The input of a reconfiguration problem consists of two \textit{feasible solutions} of some \textit{source problem} (e.g., \textsc{Satisfiability}, \textsc{Independent Set}, \textsc{Vertex Cover}, \textsc{Dominating Set}, etc.) and a \emph{reconfiguration rule} that describes an adjacency relation between solutions.
One of the primary goal is to decide whether one feasible solution can be transformed into the other via a sequence of adjacent feasible solutions where each intermediate member is obtained from its predecessor by applying the given reconfiguration rule exactly once.
Such a sequence, if exists, is called a \emph{reconfiguration sequence}.
In general, the (huge) set of \textit{all} feasible solutions is \textit{not} given, and one can verify in polynomial time (with respect to the input's size) the feasibility and adjacency of solutions.
Another way of looking at reconfiguration problems involves the so-called \textit{reconfiguration graphs}---a graph whose vertices/nodes are feasible solutions and their adjacency relation is defined via the given reconfiguration rule.
Naturally, a reconfiguration sequence is nothing but a walk in the corresponding reconfiguration graph.
Readers may recall the classic Rubik's cube puzzle as an example of a reconfiguration problem, where each configuration of the Rubik's cube corresponds to a feasible solution, and two configurations (solutions) are adjacent if one can be obtained from the other by rotating a face of the cube by either $90$, $180$, or $270$ degrees.
The question is whether one can transform an arbitrary configuration to the one where each face of the cube has only one color.
For an overview of this research area, we refer readers to the surveys~\cite{Heuvel13,Nishimura18,MynhardtN19,BousquetMNS22}.

Reconfiguration problems involving \RR{\textit{vertex-subsets} }(e.g., clique, independent set, vertex cover, dominating set, etc.) of a graph have been extensively considered in the literature.
In such problems, to make it more convenient for describing reconfiguration rules, one usually views a \RR{vertex-subset} of a graph as a set of tokens placed on its vertices.
Some well-known reconfiguration rules in this setting are:
\begin{itemize}
	\item Token Sliding ($\sfTS$): a token can only move to one of the unoccupied adjacent vertices;
	\item Token Jumping ($\sfTJ$): a token can move to any unoccupied vertex; and
	\item Token Addition/Removal ($\sfTAR(k)$): a token can either be added to an unoccupied vertex or \RR{removed} from an occupied one, such that the number of tokens is always (upper/lower) bounded by some given positive integer $k$.
\end{itemize}

Let $G$ be a simple, undirected graph.
An \textit{independent set} of a graph $G$ is a set of pairwise non-adjacent vertices.
The \textsc{Independent Set (\RR{IS})} problem, which asks if $G$ has an independent set of size at least some given positive integer $k$, is one of the fundamental $\ttNP$-complete problems in the computational complexity theory~\cite{GareyJohnson79}.
Given an integer $d \geq 2$, a \textit{distance-$d$ independent set} (also known as \textit{$d$-scattered set} 
or sometimes \textit{$d$-independent set}%
\footnote{In fact, the terminology \textit{$d$-independent set} is sometimes used to indicate a \RR{vertex-subset} such that any two members are of distance at least $d+1$ for $d \geq 1$. 
	We note that sometimes a \textit{$d$-independent set} is defined as a \RR{vertex-subset} $S$ such that the maximum degree of the subgraph induced by $S$ is at most $d$. 
	In some other contexts, a \textit{$d$-independent set} is nothing but an independent set of size $d$.}%
) of $G$ is a set of vertices whose pairwise distance is at least $d$.
This ``distance-$d$ independent set'' concept generalizes ``independent set'': an independent set is also a distance-$2$ independent set but may not be a distance-$d$ independent set for $d \geq 3$.
Given a fixed integer $d \geq 2$, the \textsc{Distance-$d$ Independent Set (\RR{D$d$IS})} problem asks if there is a distance-$d$ independent set of $G$ whose size is at least some given positive integer $k$.
Clearly, \textsc{\RR{D$2$IS}} is nothing but \textsc{\RR{IS}}.
It is known that \textsc{\RR{D$d$IS}} is $\ttNP$-complete for every fixed $d \geq 3$ on general graphs~\cite{KongZ93} and even on regular bipartite graphs when $d \in \{3,4,5\}$~\cite{KongZ00}.
Eto~et~al.~\cite{EtoGM14} proved that \textsc{\RR{D$d$IS}} is $\ttNP$-complete for every fixed $d \geq 3$ even for planar bipartite graphs of maximum degree three. 
They also proved that on chordal graphs, \textsc{\RR{D$d$IS}} is in $\ttP$ for any even $d \geq 2$ and remains $\ttNP$-complete for any odd $d \geq 3$.
The complexity of \textsc{\RR{D$d$IS}} on several other graphs has also been studied~\cite{AgnarssonDH03,MontealegreT16,JenaJDN18}.
Additionally, \textsc{\RR{D$d$IS}} and its variants have been studied extensively from different viewpoints, including exact exponential algorithms~\cite{YamanakaKT19}, approximability~\cite{KatsikarelisLP19,EtoGM14,JenaJDN18}, and parameterized complexity~\cite{KatsikarelisLP20}.

In this paper, we take \textsc{\RR{D$d$IS}} as the source problem and initiate the study of \textsc{Distance-$d$ Independent Set Reconfiguration (D$d$ISR)} from the \RR{classic} complexity viewpoint.
The problem for $d = 2$, which is usually known as the \textsc{Independent Set Reconfiguration (ISR)} problem, has been well-studied from both classic and parameterized complexity viewpoints.
\textsc{ISR} under $\sfTS$ was first introduced by Hearn and Demaine~\cite{HearnD05}.
Ito~et~al.~\cite{ItoDHPSUU11} introduced and studied the problem under $\sfTAR$ along with several other reconfiguration problems.
Under $\sfTJ$, the problem was first studied by Kami{\'n}ski~et~al.~\cite{KaminskiMM12}.
Very recently, the study of \textsc{ISR} on directed graphs has been initiated by Ito~et~al.~\cite{ItoIKNOTW22}.
Additionally, research on the structure and realizability of reconfiguration graphs for \textsc{ISR} has been initiated both under $\sfTS$~\cite{AvisH23-1,AvisH23-2} and $\sfTJ$~\cite{BousquetDPT23}.
Readers are referred to~\cite{Nishimura18,BousquetMNS22} for a complete overview of recent developments regarding \textsc{ISR}.
We now briefly mention some known results regarding the computational complexity of \textsc{ISR} on different graph classes.
\textsc{ISR} remains $\ttPSPACE$-complete under any of $\sfTS, \sfTJ, \sfTAR$ on general graphs~\cite{ItoDHPSUU11}, planar graphs of maximum degree three and bounded bandwidth~\cite{HearnD05,Zanden15,Wrochna18} and perfect graphs~\cite{KaminskiMM12}.
Under $\sfTS$, the problem is $\ttPSPACE$-complete even on split graphs~\cite{BelmonteKLMOS21}.
Interestingly, on bipartite graphs, \textsc{ISR} under $\sfTS$ is $\ttPSPACE$-complete while under any of $\sfTJ$ and $\sfTAR$ it is $\ttNP$-complete~\cite{LokshtanovM19}.
On the positive side, \textsc{ISR} under any of $\sfTJ$ and $\sfTAR$ is in $\ttP$ on even-hole-free graphs~\cite{KaminskiMM12} (which also contains chordal graphs, split graphs, interval graphs, trees, etc.), cographs~\cite{Bonsma16}, and claw-free graphs~\cite{BonsmaKW14}.
\textsc{ISR} under $\sfTS$ is in $\ttP$ on cographs~\cite{KaminskiMM12}, claw-free graphs~\cite{BonsmaKW14}, trees~\cite{DemaineDFHIOOUY15}, bipartite permutation graphs and bipartite distance-hereditary graphs~\cite{Fox-EpsteinHOU15}, and interval graphs~\cite{BonamyB17,BrianskiFHM21}.
\RA{Very recently, Bartier et al.~\cite{BartierBM24} proved that \textsc{ISR} under any of $\sfTS$ and $\sfTJ$ is $\ttPSPACE$-complete on $H$-free graphs when $H$ is neither a path nor a subdivision of the claw.
Additionally, they designed a polynomial-time algorithm to solve \textsc{ISR} under $\sfTS$ when the input graph is fork-free.}

\textsc{D$d$ISR} for $d \geq 3$ was first studied by Siebertz~\cite{Siebertz18} from the parameterized complexity viewpoint.
More precisely, in~\cite{Siebertz18}, Siebertz proved that \textsc{D$d$ISR} under $\sfTAR$ is in $\ttFPT$ for every $d \geq 2$ on ``nowhere dense graphs'' (which generalized the previously known result for $d = 2$ of Lokshtanov~et~al.~\cite{LokshtanovMPRS15}) and it is $\mathtt{W[1]}$-hard for some value of $d \geq 2$ on ``somewhere dense graphs'' that are closed under taking subgraphs.

Since the $\sfTJ$ and $\sfTAR$ rules are somewhat equivalent~\cite{KaminskiMM12}, i.e., any $\sfTJ$-sequence between two size-$k$ token-sets can be converted into a $\sfTAR$-sequence between them whose members are token-sets of size at least $k-1$ and vice versa\footnote{They proved the result for \textsc{ISR}, but it is not hard to extend it for \textsc{D$d$ISR}.}, in this paper, we consider \textsc{D$d$ISR} ($d \geq 3$) under only $\sfTS$ and $\sfTJ$.
Most of our results are summarized in Table~\ref{table:our-results}.
In short, we show the following results.
\begin{itemize}
	\item It is worth noting that
	\begin{itemize}
		\item Even though it is well-known that \textsc{\RR{D$d$IS}} on $G$ and \textsc{\RR{IS}} on its $(d-1)$th power (this concept will be defined later) are equivalent, this does \textit{not} necessarily hold for their reconfiguration variants. (Section~\ref{sec:graphs-and-their-powers}.)
		
		\item The definition of D$d$IS implies the triviality of \textsc{D$d$ISR} for large enough $d$ on graphs whose (connected) components' diameters are bounded by some constant, including cographs and split graphs.
		(Section~\ref{sec:bounded-diameter}.)
	\end{itemize}
	
	\item On chordal graphs and split graphs, there are some interesting complexity dichotomies. 
	(Section~\ref{sec:chordal-split}.)
	\begin{itemize}
		\item Under $\sfTJ$ on chordal graphs, \textsc{D$d$ISR} is in $\ttP$ for even $d$ and $\ttPSPACE$-complete for odd $d$.
		\item On split graphs, \textsc{D$d$ISR} under $\sfTS$ is $\ttPSPACE$-complete for $d = 2$~\cite{BelmonteKLMOS21} but in $\ttP$ for $d = 3$. 
		Under $\sfTJ$, it is in $\ttP$ for $d = 2$~\cite{KaminskiMM12} but $\ttPSPACE$-complete for $d = 3$.
	\end{itemize}

	\item 
	\RR{The known results for $d = 2$ on perfect graphs~\cite{KaminskiMM12} and on planar graphs of maximum degree three and bounded bandwidth~\cite{HearnD05,Zanden15,Wrochna18} can be extended for $d \geq 3$.
	(Section~\ref{sec:extend-results}.)}
\end{itemize}

\begin{table}[!ht]
	\centering
	\begin{adjustbox}{max width=\textwidth}
	\begin{tabular}{|c|c|c|c|c|}
		\hline
		\multirow{2}{*}{Graph} & \multicolumn{2}{c|}{{$\sfTS$}} & \multicolumn{2}{c|}{{$\sfTJ$}} \\
		\cline{2-5}
		& $d = 2$ & $d \geq 3$ & $d = 2$ & $d \geq 3$\\
		\hline
		\multirow{4}{*}{chordal} & \multirow{4}{*}{$\ttPSPACE$-C~\cite{BelmonteKLMOS21}} & \multirow{4}{*}{unknown} & \multirow{4}{*}{$\ttP$~\cite{KaminskiMM12}} & even $d$: $\ttP$\\
		& & & & (\cref{cor:chordal-TJ-even-d})\\
		\cline{5-5}
		& & & & odd $d$: $\ttPSPACE$-C\\
		& & & & (\cref{thm:chordal-TJ-odd-d})\\
		\hline
		& \multirow{4}{*}{$\ttPSPACE$-C~\cite{BelmonteKLMOS21}} & & \multirow{4}{*}{$\ttP$~\cite{KaminskiMM12}} & $d = 3$: $\ttPSPACE$-C\\
		split & & $\ttP$ & & (\cref{cor:split-TJ})\\
		\cline{5-5}
		($\subseteq$ chordal) & & \multirow{1}{*}{(\cref{prop:split-TS,prop:bounded-diameter})} & & $d \geq 4$: $\ttP$\\
		& & & & (\cref{prop:bounded-diameter})\\
		\hline
		tree ($\subseteq$ chordal) & $\ttP$~\cite{DemaineDFHIOOUY15} & unknown & $\ttP$~\cite{KaminskiMM12} & $\ttP$ (\cref{cor:interval-TJ})\\
		\hline
		\multirow{2}{*}{perfect} & \multirow{2}{*}{$\ttPSPACE$-C~\cite{KaminskiMM12}} & $\ttPSPACE$-C & \multirow{2}{*}{$\ttPSPACE$-C~\cite{KaminskiMM12}} & $\ttPSPACE$-C \\
		& & (\cref{thm:perfect-TSTJ}) & & (\cref{thm:perfect-TSTJ})\\
		\hline
		\multirow{2}{*}{\begin{tabular}{@{}c@{}}planar $\cap$ subcubic\\ $\cap$ bounded bandwidth\end{tabular}} & $\ttPSPACE$-C & $\ttPSPACE$-C  & $\ttPSPACE$-C & $\ttPSPACE$-C \\
		& \cite{HearnD05,Zanden15,Wrochna18} & (\cref{thm:planar-TSTJ}) & \cite{HearnD05,Zanden15,Wrochna18} & (\cref{thm:planar-TSTJ})\\
		\hline
		cograph ($\subseteq$ perfect) & $\ttP$~\cite{KaminskiMM12} & $\ttP$ (\cref{cor:cographs-TSTJ}) & $\ttP$~\cite{Bonsma16} & $\ttP$ (\cref{cor:cographs-TSTJ})\\
		\hline
		interval ($\subseteq$ chordal) & $\ttP$~\cite{BonamyB17,BrianskiFHM21} & unknown & $\ttP$~\cite{KaminskiMM12} & $\ttP$ (\cref{cor:interval-TJ})\\
		\hline
	\end{tabular}
	\end{adjustbox}
	\caption{\RR{The computational complexity of \textsc{D$d$ISR} ($d \geq 3$) on some graphs considered in this paper. For comparison, we also mention the corresponding known results for $d = 2$. $\ttPSPACE$-C stands for $\ttPSPACE$-complete.}}
	\label{table:our-results}
\end{table}

\section{Preliminaries}%
\label{sec:preliminaries}

For terminology and notation not defined here, see~\cite{Diestel17}.
Let $G$ be a simple, undirected graph with vertex-set $V(G)$ and edge-set $E(G)$. 
For two sets $I, J$, we sometimes use $I - J$ and $I + J$ to indicate $I \setminus J$ and $I \cup J$, respectively.
Additionally, we simply write $I - u$ and $I + u$ instead of $I - \{u\}$ and $I + \{u\}$, respectively. 
The \textit{neighborhood} of a vertex $v$ in $G$, denoted by $N_G(v)$, is the set $\{w \in V(G): vw \in E(G)\}$.
The \textit{closed neighborhood} of $v$ in $G$, denoted by $N_G[v]$, is simply the set $N_G(v) + v$.
Similarly, for a \RR{vertex-subset} $I \subseteq V(G)$, its \textit{neighborhood} $N_G(I)$ and \textit{closed neighborhood} $N_G[I]$ are respectively $\bigcup_{v \in I}N_G(v)$ and $N_G(I) + I$.
The \textit{degree} of a vertex $v$ in $G$, denoted by $\deg_G(v)$, is $\vert N_G(v) \vert$.
The \textit{distance} between two vertices $u, v$ in $G$, denoted by $\dist_G(u, v)$, is the number of edges in a shortest path between them.
For convenience, if there is no path between $u$ and $v$ then $\dist_G(u, v) = \infty$.
The \textit{diameter} of $G$, denoted by $\diam(G)$, is the largest distance between any two vertices.
A \textit{(connected) component} of $G$ is a maximal subgraph in which there is a path connecting any pair of vertices.
An \textit{independent set (IS)} of $G$ is a \RR{vertex-subset} \RR{$I \subseteq V(G)$} such that for any $u, v \in I$, we have $uv \notin E(G)$.
A \textit{distance-$d$ independent set (D$d$IS)} of $G$ for an integer $d \geq 2$ is a \RR{vertex-subset} $I \subseteq V(G)$ such that for any $u, v \in I$, $\dist_G(u, v) \geq d$. 

Unless otherwise noted, we denote by \RR{$(G, I, J, \sfR)$} an instance of \textsc{D$d$ISR} under $\sfR \in \{\sfTS, \sfTJ\}$ where $I$ and $J$ are two distinct D$d$ISs of a given graph $G$, for some fixed $d \geq 2$.
Since \RR{$(G, I, J, \sfR)$} is obviously a no-instance if $\vert I \vert \neq \vert J \vert$, from now on, we always assume that $\vert I \vert = \vert J \vert$.
Imagine that a token is placed on each vertex in a D$d$IS of a graph $G$. 
A \textit{$\sfTS$-sequence} in $G$ between two D$d$ISs $I$ and $J$ is the sequence $\mathcal{S} = \langle I = I_0, I_1, \dots, I_q = J \rangle$ such that for $i \in \{0, \dots, q-1\}$, the set $I_i$ is a D$d$IS of $G$ and there exists a pair $x_i, y_i \in V(G)$ such that $I_i - I_{i+1} = \{x_i\}$, $I_{i+1} - I_i = \{y_i\}$, and $x_iy_i \in E(G)$.
By simply removing the restriction $x_iy_i \in E(G)$, we immediately obtain the definition of a \textit{$\sfTJ$-sequence} in $G$.
Depending on the considered rule $\sfR \in \{\sfTS, \sfTJ\}$, we can also say that $I_{i+1}$ is obtained from $I_i$ by \textit{immediately sliding/jumping} a token from $x_i$ to $y_i$ and write $x_i \reconf[\sfR]{G} y_i$.
As a result, we can also write $\mathcal{S} = \langle x_0 \reconf[\sfR]{G} y_0, \dots, x_{q-1} \reconf[\sfR]{G} y_{q-1} \rangle$.
In short, $\mathcal{S}$ can be viewed as a (ordered) sequence of either D$d$ISs or token-moves.
(Recall that we defined $\mathcal{S}$ as a sequence between $I$ and $J$. As a result, when regarding $\mathcal{S}$ as a sequence of token-moves, we implicitly assume that the initial D$d$IS is $I$.) 
With respect to the latter viewpoint, we say that $\mathcal{S}$ \textit{slides/jumps a token $t$ from $u$ to $v$ in $G$} if $t$ is originally placed on $u \in I_0$ and finally on $v \in I_q$ after performing $\mathcal{S}$. 
The \textit{length} of a $\sfR$-sequence is simply the number of times the rule $\sfR$ is applied.
\RR{In particular, the $\sfR$-sequence $\mathcal{S} = \langle x_0 \reconf[\sfR]{G} y_0, \dots, x_{q-1} \reconf[\sfR]{G} y_{q-1} \rangle$ has length $q$.}

We conclude this section with the following remark: since \textsc{\RR{D$d$IS}} is in $\ttNP$, \textsc{D$d$ISR} is always in $\ttPSPACE$~\cite{ItoDHPSUU11}. 
As a result, to show the $\ttPSPACE$-completeness of \textsc{D$d$ISR}, it is sufficient to construct a polynomial-time reduction from a known $\ttPSPACE$-hard problem and prove its correctness.

\section{Observations}%
\label{sec:observations}

\subsection{Graphs and Their Powers}%
\label{sec:graphs-and-their-powers}

An extremely useful concept for studying distance-$d$ independent sets is the so-called \textit{graph power}.
For a graph $G$ and an integer $s \geq 1$, the \textit{$s$th power of $G$} is the graph $G^s$ whose vertices are $V(G)$ and two vertices $u, v$ are adjacent in $G^s$ if $\dist_G(u, v) \leq s$. 
Observe that $I$ is a distance-$d$ independent set of $G$ if and only if $I$ is an independent set of $G^{d-1}$.
Therefore, \textsc{\RR{D$d$IS}} in $G$ is equivalent to \textsc{\RR{IS}} in $G^{d-1}$.

However, this may \textit{not} apply for their reconfiguration variants.
More precisely, the statement ``the \textsc{D$d$ISR}'s instance \RR{$(G, I, J, \sfR)$} is a yes-instance if and only if the \textsc{ISR}'s instance \RR{$(G^{d-1}, I, J, \sfR)$} is a yes-instance'' holds for $\sfR = \sfTJ$ but not for $\sfR = \sfTS$.
The reason is that under $\sfTJ$ we do not care about edges when performing token-jumps (as long as they result new D$d$ISs), therefore whatever token-jump we perform in $G$ can also be done in $G^{d-1}$ and vice versa.
Therefore, we have \RR{the following proposition}.
\begin{proposition}\label{prop:results-inherit-from-ISR-TJ}
	Let $\calG$ and $\calH$ be two graph classes and suppose that for every $G \in \calG$ we have $G^{d-1} \in \calH$ for some fixed integer $d \geq 2$.
	If \textsc{ISR} under $\sfTJ$ on $\calH$ can be solved in polynomial time, so does \textsc{D$d$ISR} under $\sfTJ$ on $\calG$.
\end{proposition}
Recall that the power of any interval graph is also an interval graph~\cite{AgnarssonGH00,ChenC10} and 
\textsc{ISR} under $\sfTJ$ on even-hole-free graphs (which contains interval graphs) is in $\ttP$~\cite{KaminskiMM12}.
Along with Proposition~\ref{prop:results-inherit-from-ISR-TJ}, we immediately obtain the following corollary.
\begin{corollary}\label{cor:interval-TJ}
	\textsc{D$d$ISR} under $\sfTJ$ is in $\ttP$ on interval graphs for any $d \geq 2$.
\end{corollary}

On the other hand, when using token-slides, we need to consider which edges can be used for moving tokens, and clearly $G^{d-1}$ has much more edges than $G$, which means certain token-slides we perform in $G^{d-1}$ cannot be done in $G$.
More precisely, we have \RR{the following proposition}.
\begin{proposition}\label{prop:TS-power-graph}
	Let $d \geq 3$ be a fixed integer.
	There exists a \textsc{D$d$ISR}'s no-instance \RR{$(G, I, J, \sfTS)$} such that the \textsc{ISR}'s instance \RR{$(G^{d-1}, I, J, \sfTS)$} is a yes-instance, where $I, J$ are D$d$ISs of $G$ (and therefore ISs of $G^{d-1}$) of size $k \geq 2$.
\end{proposition}
\begin{proof}
	We construct a graph $G$ as follows.
	\begin{figure}[!ht]
		\centering
		\begin{adjustbox}{max width=0.5\textwidth}
		\begin{tikzpicture}[every node/.style={circle, draw, thick, fill=white, minimum size=8mm}]

			\node[fill=gray!30!white, draw=none, rectangle, minimum width=1cm, minimum height=11cm, label=above:$P_1$] at (0, -4) {};
			\node[fill=gray!30!white, draw=none, rectangle, minimum width=1cm, minimum height=11cm, label=above:$P_2$] at (2, -4) {};
			\node[fill=gray!30!white, draw=none, rectangle, minimum width=1cm, minimum height=11cm, label=above:$P_k$] at (6, -4) {};

			\node[label=above:$v_1$] (v1) at (0,0) {};
			\node[label=above:$v_2$] (v2) at (2,0) {};
			\node[label=above:$v_k$] (v3) at (6,0) {};
			\node[label={below right:$v_1^\star$}] (vs1) at (0,-2) {};
			\node[label={below right:$v_2^\star$}] (vs2) at (2,-2) {};
			\node[label={below right:$v_k^\star$}] (vs3) at (6,-2) {};
			\node[label=below:$w_1$] (w1) at (0,-8) {};
			\node[label=below:$w_2$] (w2) at (2,-8) {};
			\node[label=below:$w_k$] (w3) at (6,-8) {};
			\node[label={above right:$w_1^\star$}] (ws1) at (0,-6) {};
			\node[label={above right:$w_2^\star$}] (ws2) at (2,-6) {};
			\node[label={above right:$w_k^\star$}] (ws3) at (6,-6) {};
			
			\node[minimum size=6mm, fill=black] at (v1) {};
			\node[minimum size=6mm, fill=black] at (v2) {};
			\node[minimum size=6mm, fill=black] at (v3) {};
			\node[minimum size=6mm, fill=gray] at (w1) {};
			\node[minimum size=6mm, fill=gray] at (w2) {};
			\node[minimum size=6mm, fill=gray] at (w3) {};
			
			\draw[dotted, thick] ([shift={(0.3,0)}]v2.east) -- ([shift={(-0.3,0)}]v3.west);
			\draw[dotted, thick] ([shift={(0.3,0)}]vs2.east) -- ([shift={(-0.3,0)}]vs3.west);
			\draw[dotted, thick] ([shift={(0.3,0)}]w2.east) -- ([shift={(-0.3,0)}]w3.west);
			\draw[dotted, thick] ([shift={(0.3,0)}]ws2.east) -- ([shift={(-0.3,0)}]ws3.west);
			\foreach \i in {1,2,3} {
				\foreach \j in {1,2,3} {
					\ifthenelse{\i=\j}{
						\draw[thick] (v\i) -- (vs\i);
						\draw[thick] (w\i) -- (ws\i);
						\draw[dashed, thick] (vs\i) -- (ws\i);
					}{
						\draw[dashed, ultra thick] (v\i) -- (vs\j);
						\draw[dashed, ultra thick] (w\i) -- (ws\j);
					}
				}
			}
		\end{tikzpicture}
		\end{adjustbox}
		\caption{\RR{Construction of a \textsc{D$d$ISR}'s instance \RR{$(G, I, J, \sfTS)$} satisfying Proposition~\ref{prop:TS-power-graph}. 
		Paths $P_1, \dots, P_k$ are marked by light gray boxes. 
		Paths $Q_{v_i^\star, v_j}$ and $Q_{w_i^\star, w_j}$ of length $d-1$ are represented by bold dashed egdes, where $1 \leq i, j \leq k$ and $i \neq j$. 
		Tokens in $I$ (resp. $J$) are marked with black (resp. gray) color.}}
		\label{fig:TS-power-graph}
	\end{figure}
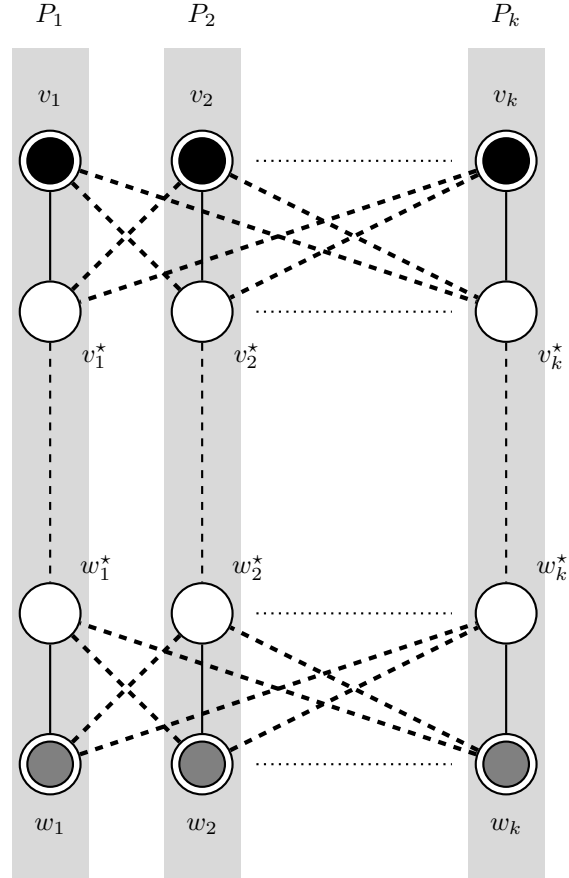
	Let $P_1, \dots, P_k$ be $k$ disjoint paths of length $d - 1$, for some $k \geq 2$.
	For each $i \in \{1, 2, \dots, k\}$, let $v_i$ and $w_i$ be the endpoints of $P_i$ and let $v^\star_i$ (resp. $w^\star_i$) be the unique neighbor of $v_i$ (resp. $w_i$) in $P_i$.
	For each $i \in \{1, 2, \dots, k\}$, join $v^\star_i$ and $v_j$ (resp. $w^\star_i$ and $w_j$), where $j \neq i$ and $1 \leq j \leq k$, by a path $Q_{v^\star_i, v_j}$ (resp. $Q_{w^\star_i, w_j}$) of length $d-1$.
	Let $I = \{v_1, v_2, \dots, v_k\}$, $I^\star = \{v^\star_1, v^\star_2, \dots, v^\star_k\}$, $J = \{w_1, w_2, \dots, w_k\}$, and $J^\star = \{w^\star_1, w^\star_2, \dots, w^\star_k\}$.
	(See \figurename~\ref{fig:TS-power-graph}.)
	
	We first verify that $I$ is indeed a D$d$IS of $G$.
	Similar arguments also hold for $J$.
	Let $v_i, v_j$ be any two distinct members of $I$, where $1 \leq i < j \leq k$.
	From the construction of $G$, one can verify that any path $P^{ij}$ between $v_i$ and $v_j$ is of length
	\begin{itemize}
		\item $d$ if either $V(P^{ij}) \cap I^\star = \{v^\star_i\}$ or $V(P^{ij}) \cap I^\star = \{v^\star_j\}$;
		\item $2(d-1)$ if $\RR{V(P^{ij})} \cap I^\star = \{v^\star_\ell\}$ for some $\ell \in \{1, 2, \dots, k\} - \{i, j\}$;
		\item at least $2(d-3) + 4$ in the remaining cases, since any other path $P^{ij}$ must contain one path of length at least $1$ between $v_i$ and some $v^\star_p$, one path of length at least $1$ between $v_j$ and some $v^\star_q \neq v^\star_p$, exactly two paths of length $d - 3$: one between $v^\star_p$ and $w^\star_p$ and the other between $v^\star_q$ and $w^\star_q$, and one path of length at least $2$ between $w^\star_p$ and $w^\star_q$,  for some $p, q \in \{1, 2, \dots, k\}$.
	\end{itemize}
	Thus, when $d \geq 3$, any path $P^{ij}$ between $v_i$ and $v_j$ has length at least $d$, and the equality happens when either $V(P^{ij}) \cap I^\star = \{v^\star_i\}$ or $V(P^{ij}) \cap I^\star = \{v^\star_j\}$.
	It follows that $\dist_G(v_i, v_j) \geq d$.
	Since this holds for any pair $v_i, v_j \in I$ ($1 \leq i < j \leq k$), the set $I$ is a D$d$IS of $G$.
	
	Next, we prove that \RR{$(G, I, J, \sfTS)$} is a no-instance of \textsc{D$d$ISR}.
	To show this, we prove that no tokens in $I$ can be moved even to one of its adjacent vertices and (by applying similar arguments) so do those in $J$.
	\RR{For each $i \in \{1, 2, \dots, k\}$, let $x_{ij}$ be the unique neighbor of $v_i$ in the path $Q_{v^\star_j, v_i}$ where $j \in \{1, 2, \dots, k\} - i$.}
	From the construction of $G$, $N_G(v_i) = \bigcup_{j \in \{1, 2, \dots, k\} - i}\{x_{ij}\} \cup \{v^\star_i\}$.
	We now prove that for each $i \in \{1, 2, \dots, k\}$, a token $t_i$ on $v_i$ can be slid to one of its neighbors if all other tokens $t_j$ on $v_j$, for $j \in \{1, 2, \dots, k\} - i$, must be moved before moving $t_i$.
	If this claim holds, no tokens in $I$ (and similarly in $J$) can be moved, and therefore \RR{$(G, I, J, \sfTS)$} is a no-instance.
	
	If $t_i$ is moved to $v^\star_i$, since all $Q_{v^\star_i, v_j}$ are of length $d-1$, the tokens $t_j$ must be moved before $t_i$, for $j \in \{1, 2, \dots, k\} - i$.
	If $t_i$ is moved to some $x_{ij}$, since $Q_{v^\star_j, v_i}$ is of length $d - 1$, the token $t_j$ on $v_j$ must be moved before $t_i$, for $j \in \{1, 2, \dots, k\} - 1$.
	By applying the above arguments for $t_j$ instead of $t_i$, we obtain that there must be some token $t_\ell$, where $\ell \in \{1, 2, \dots, k\} - \{i, j\}$ that must be moved before $t_j$, which clearly also before $t_i$.
	By repeatedly applying these arguments, we finally obtain that all other tokens must be moved before $t_i$.
	
	It remains to show that \RR{$(G^{d-1}, I, J, \sfTS)$} is a yes-instance of \textsc{ISR}.
	Observe that for each $i \in \{1, 2, \dots, k\}$, we have $\dist_G(v_i, w_i) = d-1$ and $\dist_G(v_i, w_j) \geq 2d - 3 \geq d$ for every $d \geq 3$ and $j \in \{1, 2, \dots, k\} - i$.
	Therefore, in $G^{d-1}$, there is an edge between $v_i$ and $w_i$, while there is no edge joining $v_i$ and $w_j$, for $i, j \in \{1, 2, \dots, k\}$ and $i \neq j$.
	Thus, $\mathcal{S} = \langle v_1 \reconf[\sfTS]{G^{d-1}} w_1, v_2 \reconf[\sfTS]{G^{d-1}} w_2, \dots, v_k \reconf[\sfTS]{G^{d-1}} w_k \rangle$ is a $\sfTS$-sequence that transforms $I$ into $J$.
\end{proof}

\subsection{Graphs With Bounded Diameter Components}%
\label{sec:bounded-diameter}

The following observation is straightforward.
\begin{proposition}\label{prop:bounded-diameter}
	Let $\calG$ be a graph class such that there is some constant $c > 0$ satisfying $\diam(C_G) \leq c$ for any $G \in \calG$ and any component $C_G$ of $G$.
	Then, \textsc{D$d$ISR} on $\calG$ under $\sfR \in \{\sfTS, \sfTJ\}$ is in $\ttP$ for every $d \geq c+1$.
\end{proposition}
\begin{proof}
	When $d \geq c+1$, any D$d$IS contains at most one vertex in each component of $G$, and the problem becomes trivial: under $\sfTJ$, the answer is always ``yes''; under $\sfTS$, compare the number of tokens in each component.
\end{proof}

As a result, on cographs (a.k.a $P_4$-free graphs), one can immediately derive the following corollary.
\begin{corollary}\label{cor:cographs-TSTJ}
	\textsc{D$d$ISR} on cographs under $\sfR \in \{\sfTS, \sfTJ\}$ is in $\ttP$ for any $d \geq 2$.
\end{corollary}
\begin{proof}
	It is well-known that the problems for $d = 2$ is in $\ttP$~\cite{KaminskiMM12,Bonsma16}.
	Since a connected cograph has diameter at most two, Proposition~\ref{prop:bounded-diameter} settles the case $d \geq 3$.
\end{proof}

\section{Chordal Graphs and Split Graphs}%
\label{sec:chordal-split}

In this section, we will focus on chordal graphs and split graphs.
Recall that the odd power of a chordal graph is also chordal~\cite{BalakrishnanP83,AgnarssonGH00} and 
\textsc{ISR} under $\sfTJ$ on even-hole-free graphs (which contains chordal graphs) is in $\ttP$~\cite{KaminskiMM12}.
Therefore, it follows from Proposition~\ref{prop:results-inherit-from-ISR-TJ} that \RR{the following corollary holds}.
\begin{corollary}\label{cor:chordal-TJ-even-d}
	\textsc{D$d$ISR} is in $\ttP$ on chordal graphs under $\sfTJ$ for any even $d \geq 2$.
\end{corollary}

In contrast, we have the following theorem.
\begin{theorem}\label{thm:chordal-TJ-odd-d}
	\textsc{D$d$ISR} is $\ttPSPACE$-complete on chordal graphs under $\sfTJ$ for any odd $d \geq 3$.
\end{theorem}
\begin{proof}
	We reduce from the \textsc{ISR} problem, which is known to be $\ttPSPACE$-complete under $\sfTJ$~\cite{ItoDHPSUU11}.
	Let \RR{$(G, I, J, \sfTJ)$} be an \textsc{ISR}'s instance.
	We construct a \textsc{D$d$ISR}'s instance \RR{$(\Gp, \Ip, \Jp, \sfTJ)$} ($d \geq 3$ is odd) as follows.
	We note that a similar reduction was used by Eto~et~al.~\cite{EtoGM14} for showing the $\ttNP$-completeness of \textsc{\RR{D$d$IS}} on chordal graphs for odd $d \geq 3$.
	We first describe how to construct $\Gp$ from $G$.
	\RR{For each vertex $v \in V(G)$, we add $v$ to $V(\Gp)$.
	For each edge $uv \in V(G)$, we add a vertex $x_{uv}$ to $V(\Gp)$ and create an edge in $\Gp$ between $x_{uv}$ and both $u$ and $v$.
	Next, we add an edge in $G^\prime$ between $x_{uv}$ and $x_{u^\prime v^\prime}$ for any pair of distinct edges $uv, u^\prime v^\prime \in E(G)$.
	Finally, for each $v \in V(G)$, we add to $\Gp$ a new path $P_v$ on $(d-3)/2$ vertices and then add an edge in $\Gp$ between $v$ and one of $P_v$'s endpoints.}
	One can verify that $\Gp$ is indeed a chordal graph: it is obtained from a split graph by attaching new paths to certain vertices.
	Clearly this construction can be done in polynomial time.
	(For example, see \figurename~\ref{fig:reduce-ISR-DdISR-chordal}.)
	
	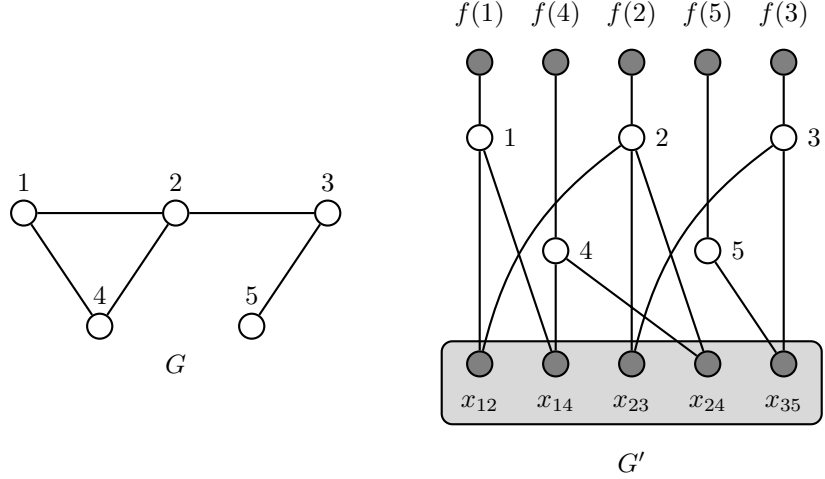
\begin{figure}[!ht]
		\centering
		\begin{adjustbox}{max width=\textwidth}
			\begin{tikzpicture}[every node/.style={circle, draw, thick, fill=white, minimum size=3mm}]
				\begin{scope}
					\foreach \i/\x/\y in {1/0/0,2/2/0,3/4/0,4/1/-1.5,5/3/-1.5} {
						\node[label={[label distance=-0.1cm]above:$\i$}] (\i) at (\x, \y) {};
					}
					\draw[thick] (1) -- (2) -- (3)-- (5) (1) -- (4) -- (2);
					\node[rectangle, draw=none, fill=none] at (2,-2) {$G$};
				\end{scope}
				\begin{scope}[shift={(6,1)}]
					\foreach \i/\x/\y in {1/0/0,2/2/0,3/4/0,4/1/-1.5,5/3/-1.5} {
						\node[label={[label distance=-0.1cm]right:$\i$}] (\i) at (\x, \y) {};
					}
					\foreach \i/\x/\y in {12/0/-3,14/1/-3,23/2/-3,24/3/-3,35/4/-3} {
						\node[fill=gray, label={[label distance=-0.1cm]below:$x_{\i}$}] (\i) at (\x, \y) {};
					}
					\foreach \i/\x/\y in {1/0/1,4/1/1,2/2/1,5/3/1,3/4/1} {
						\node[fill=gray, label={[label distance=-0.1cm]above:$f(\i)$}] (P\i) at (\x, \y) {};
						\draw[thick] (\i) -- (P\i);
					}
					\draw[thick] (12) -- (1) (12) edge[bend left=20] (2) (14) -- (1) (14) -- (4) (23) -- (2) (23) edge[bend left=20] (3) (24) -- (2) (24) -- (4) (35) -- (3) (35) -- (5);
					
					\begin{scope}[on background layer]
						\node[rectangle, draw, thick, fill=gray!30!white, rounded corners, minimum width=5cm, minimum height=1.1cm] at ([shift={(0,-0.25)}]23) {};
					\end{scope}
					
					\node[rectangle, draw=none, fill=none] at (2,-4.3) {$\Gp$};
				\end{scope}
			\end{tikzpicture}
		\end{adjustbox}
		\caption{An example of constructing a chordal graph $\Gp$ from $G$ for $d = 5$. Vertices in a light-gray box form a clique. Vertices in $V(\Gp) - V(G)$ are marked with the gray color.}
		\label{fig:reduce-ISR-DdISR-chordal}
	\end{figure}
	
	For each $u \in V(G)$, we define $f(u) \in V(\Gp)$ to be the vertex whose distance in $\Gp$ from $u$ is largest among all vertices in $V(P_u) + u$.
	Let $f(X) = \bigcup_{x \in X}\{f(x)\}$ for a \RR{vertex-subset} $X \subseteq V(G)$.
	From the construction of $\Gp$, note that if $u$ and $v$ are two vertices of distance $2$ in $G$, one can always find a shortest path $Q$ between $f(u)$ and $f(v)$ whose length is exactly $d$.
	Indeed, $Q$ can be obtained by joining the paths from $f(u)$ to $u$, from $u$ to $x_{uw}$, from $x_{uw}$ to $x_{wv}$, from $x_{wv}$ to $v$, and from $v$ to $f(v)$, where $w \in N_G(u) \cap N_G(v)$.
	It follows that if $I$ is an independent set of $G$ then $f(I)$ is a distance-$d$ independent set of $\Gp$.
	Therefore, we can set $\Ip = f(I)$ and $\Jp = f(J)$.
	
	From the construction of $\Gp$, note that for each $uv \in E(G)$, $x_{uv}$ is of distance 
	exactly one from each $x_{wz}$ for $wz \in E(G) - uv$, 
	at most two from each $v \in V(G)$, 
	and at most $2 + (d-3)/2 \leq d - 1$ from each vertex in $P_v$ for $v \in V(G)$.
	It follows that any distance-$d$ independent set of $\Gp$ of size at least two must not contain any vertex in $\bigcup_{uv \in E(G)}\{x_{uv}\}$.
	
	We now show that there is a $\sfTJ$-sequence between $I$ and $J$ in $G$ if and only if there is a $\sfTJ$-sequence between $\Ip$ and $\Jp$ in $\Gp$.
	If $\vert I \vert = \vert J \vert = 1$, the claim is trivial.
	As a result, we consider the case $\vert I \vert = \vert J \vert \geq 2$.
	Since $f(I)$ is a distance-$d$ independent set in $\Gp$ if $I$ is an independent set in $G$, the only-if direction is clear.
	It remains to show the if direction.
	Let $\calS^\prime$ be a $\sfTJ$-sequence in $\Gp$ between $\Ip$ and $\Jp$.
	We modify $\calS^\prime$ by repeating the following steps:
	\begin{itemize}
		\item Let $x \reconf[\sfTJ]{\Gp} y$ be the first token-jump that move a token from $x \in f(V(G))$ to some $y \in V(P_u) + u - f(u)$ for some $u \in V(G)$. 
		If no such token-jump exists, we stop.
		Let $I_x$ and $I_y$ be respectively the distance-$d$ independent sets obtained before and after this token-jump.
		In particular, $I_y = I_x - x + y$.
		\item Replace $x \reconf[\sfTJ]{\Gp} y$ by $x \reconf[\sfTJ]{\Gp} f(u)$ and replace the first step after $x \reconf[\sfTJ]{\Gp} y$ of the form $y \reconf[\sfTJ]{\Gp} z$ by $f(u) \reconf[\sfTJ]{\Gp} z$.
		From the construction of $\Gp$, note that any path containing $f(u)$ must also contains all vertices in $V(P_u) + u$.
		Additionally, $I_x \cap (V(P_u) + u) = \emptyset$, otherwise no token in $I_x$ can jump to $y$. 
		Therefore, the set $I_x - x + f(u)$ is also a distance-$d$ independent set.
		Moreover, $\dist_{\Gp}(f(u), z) \geq \dist_{\Gp}(y, z)$ for any $z \in V(\Gp) - V(P_u)$.
		Roughly speaking, this implies that no token-jump between $x \reconf[\sfTJ]{\Gp} y$ and $y \reconf[\sfTJ]{\Gp} z$ breaks the ``distance-$d$ restriction''.
		Thus, after the above replacements, $\calS^\prime$ is still a $\sfTJ$-sequence in $\Gp$.
		\item Repeat the first step. 
	\end{itemize}
	After modification, the final resulting $\sfTJ$-sequence $\calS^\prime$ in $\Gp$ contains only token-jumps between vertices in $f(V(G))$.
	By definition of $f$, we can construct a $\sfTJ$-sequence between $I$ and $J$ in $G$ simply by replacing each step $x \reconf[\sfTJ]{\Gp} y$ in $\calS^\prime$ by $f^{-1}(x) \reconf[\sfTJ]{G} f^{-1}(y)$.
	Our proof is complete.
\end{proof}

Now, we consider the split graphs.
Proposition~\ref{prop:bounded-diameter} implies that on split graphs (where each component has diameter at most $3$), \textsc{D$d$ISR} is in $\ttP$ under $\sfR \in \{\sfTS, \sfTJ\}$ for any $d \geq 4$.
Interestingly, recall that when $d = 2$, the problem under $\sfTS$ is $\ttPSPACE$-complete even on split graphs~\cite{BelmonteKLMOS21} while under $\sfTJ$ it is in $\ttP$~\cite{KaminskiMM12}.
It remains to consider the case $d = 3$.

Observe that the constructed graph $\Gp$ in the proof of Theorem~\ref{thm:chordal-TJ-odd-d} is indeed a split graph when $d = 3$.
Therefore, we have the following corollary.
\begin{corollary}\label{cor:split-TJ}
	\textsc{D$3$ISR} is $\ttPSPACE$-complete on split graphs under $\sfTJ$.
\end{corollary}

In contrast, under $\sfTS$, we have the following proposition.
\begin{proposition}\label{prop:split-TS}
	\textsc{D$3$ISR} is in $\ttP$ on split graphs under $\sfTS$.
\end{proposition}
\begin{proof}
	Let \RR{$(G, I, J, \sfTS)$} be an instance of \textsc{D$3$ISR} and suppose that $V(G)$ can be partitioned into a clique $K$ and an independent set $S$.
	One can assume without loss of generality that $G$ is connected, otherwise each component can be solved independently.
	If $\vert I \vert = \vert J \vert = 1$, the problem becomes trivial: \RR{$(G, I, J, \sfR)$} is always a yes-instance.
	Thus, we now consider $\vert I \vert = \vert J \vert \geq 2$.
	Observe that for every $u \in V(G)$ and $v \in K$, we have $\dist_G(u, v) \leq 2$.
	Therefore, in this case, both $I$ and $J$ are subsets of $S$.
	\RR{Note that $I \neq J$.}
	Now, no token in $I \cup J$ can be slid, otherwise such a token must be slid to some vertex in $K$, and each vertex in $K$ has distance at most two from any other token, which contradicts the restriction that tokens must form a D$3$IS.
	Hence, \RR{$(G, I, J, \sfTS)$} is always a no-instance if $\vert I \vert = \vert J \vert \geq 2$.
\end{proof}

\begin{toappendix}
\section{A Reduction under $\sfTJ$ on General Graphs}%
\label{sec:general-TJ}

Recall that Ito~et~al.~\cite{ItoDHPSUU11} proved the $\ttPSPACE$-completeness of \textsc{ISR} under $\sfTJ/\sfTAR$ by reducing from \textsc{3-Satisfiability Reconfiguration (3SAT-R)}.
In this section, we present a simple proof for the $\ttPSPACE$-hardness of \textsc{D$d$ISR} ($d \geq 3$) on general graphs under $\sfTJ$ by \textit{reducing from \textsc{ISR} instead of \textsc{3SAT-R}}.

\begin{theorem}\label{thm:general-TJ}
	\textsc{D$d$ISR} is $\ttPSPACE$-complete under $\sfTJ$ for any $d \geq 3$.
\end{theorem}
\begin{proof}
	We reduce from the \textsc{ISR} problem, which is known to be $\ttPSPACE$-complete under $\sfTJ$~\cite{ItoDHPSUU11}.
	Let \RR{$(G, I, J, \sfTJ)$} be an \textsc{ISR}'s instance.
	We construct a \textsc{D$d$ISR}'s instance \RR{$(\Gp, I, J, \sfTJ)$} ($d \geq 3$) as follows.
	Unless otherwise noted, we always assume $p = (d-1)/2$ if $d$ is odd and $p = (d - 2)/2$ if $d$ is even.
	To construct $\Gp$ from $G$, we first replace each edge $uv \in E(G)$ with a new path $P_{uv} = x_0^{uv}\dots x_{d-1}^{uv}$ of length $d-1$, where $x_0^{uv} = u$ and $x_{d-1}^{uv} = v$.
	We emphasize that the ordering of endpoints is important, i.e., $P_{vu} = x_0^{vu}\dots x_{d-1}^{vu}$ is the path where $x_i^{vu} = x_{d-1-i}^{uv}$ for $0 \leq i \leq d-1$.
	Basically, $P_{uv}$ and $P_{vu}$ describe the same path with different vertex-labels.
	If $d$ is odd, we create a new clique $K$ whose vertices are $\bigcup_{uv \in E(G)}\{x_p^{uv}\}$.
	If $d$ is even, we add a new vertex $x^\star$ and create an edge between it and every vertex in $\bigcup_{uv \in E(G)}\{x_p^{uv}, x_{p+1}^{uv}\}$.
	The resulting graph is $\Gp$.
	(For example, see \figurename~\ref{fig:reduce-ISR-DdISR}.)
	Clearly this construction can be done in polynomial time.
	
	\begin{figure}[ht]
		\centering
		\begin{adjustbox}{max width=\textwidth}
			\begin{tikzpicture}[every node/.style={circle, draw, thick, fill=white, minimum size=3mm}]
				\begin{scope}
					\foreach \i/\x/\y in {1/0/0,2/2/0,3/4/0,4/1/-1.5,5/3/-1.5} {
						\node[label={[label distance=-0.1cm]below:$\i$}] (\i) at (\x, \y) {};
					}
					\draw[thick] (1) -- (2) -- (3)-- (5) (1) -- (4) -- (2);
					\node[rectangle, draw=none, fill=none] at (2,-2.5) {$d = 2$};
				\end{scope}
				\begin{scope}[shift={(5,0.5)}]
					\foreach \i/\x/\y in {1/0/0,2/2/0,3/4/0,4/1/-1.5,5/3/-1.5} {
						\node[label={[label distance=-0.1cm]below:$\i$}] (\i) at (1.5*\x, 1.5*\y) {};
					}
					\draw[thick] (1) to[-] coordinate[pos=0.5] (x12)  (2) (1) to[-] coordinate[pos=0.5] (x14) (4) (2) to[-] coordinate[pos=0.5] (x23) (3) (2) to[-] coordinate[pos=0.5] (x24) (4) (3) to[-] coordinate[pos=0.5] (x35) (5);
					\foreach \i in {12,14,23,24,35} {
						\node[fill=gray, label={[label distance=-0.2cm]above:$x_1^{\i}$}] (\i) at (x\i) {};
					}
					\draw[thick] (12) -- (14)  (12) edge[bend left=30] (23) (12) -- (24) (12) edge[bend right=10] (35) (14) -- (24) (14) edge[bend right=30] (35) (14) edge[bend right=40] (23) (23) -- (24) (23) -- (35) (24) -- (35);
					
					\node[rectangle, draw=none, fill=none] at (1.5*2,-3) {$d = 3$};
				\end{scope}
				\begin{scope}[shift={(2,-4)}]
					\foreach \i/\x/\y in {1/0/0,2/2/0,3/4/0,4/1/-1.5,5/3/-1.5} {
						\node[label={[label distance=-0.1cm]below:$\i$}] (\i) at (2*\x, 2*\y) {};
					}
					\draw[thick] (1) to[-] coordinate[pos=0.33] (x121) coordinate[pos=0.66] (x122) (2) (1) to[-] coordinate[pos=0.33] (x141) coordinate[pos=0.66] (x142) (4) (2) to[-] coordinate[pos=0.33] (x231) coordinate[pos=0.66] (x232) (3) (2) to[-] coordinate[pos=0.33] (x241) coordinate[pos=0.66] (x242) (4) (3) to[-] coordinate[pos=0.33] (x351) coordinate[pos=0.66] (x352) (5);
					\node[rectangle, draw=none, fill=none] at (2*2,-4) {$d = 4$};
					
					\foreach \j in {1,...,2} {
						\foreach \i/\k in {12\j/12,14\j/14,23\j/23,24\j/24,35\j/35} {
							\ifthenelse{\i=142 \OR \i=242 \OR \i=352}{
								\node[fill=gray, label={[label distance=-0.1cm]below:$x_{\j}^{\k}$}] (\i) at (x\i) {};
							}{
								\node[fill=gray, label={[label distance=-0.2cm]above:$x_{\j}^{\k}$}] (\i) at (x\i) {};
							}
							
						}
					}
					\node[fill=gray, label={[label distance=-0.2cm]below:$x^\star$}] (xs) at (4,-3) {};
					\draw[thick] (121) -- (xs) (122) -- (xs) (142) edge[bend right=20] (xs) (141) -- (xs) (241) -- (xs) (242) -- (xs) (232) -- (xs) (231) -- (xs) (351) -- (xs) (352) -- (xs);
				\end{scope}
			\end{tikzpicture}
		\end{adjustbox}
		\caption{An example of constructing $\Gp$ from a given graph $G$ for some values of $d \in \{2,3,4\}$ (in case $d = 2$, we have $\Gp = G$). Vertices in $V(\Gp) - V(G)$ are marked with the gray color.}
		\label{fig:reduce-ISR-DdISR}
	\end{figure}
	
	Now, we show that an independent set of $G$ of size $k \geq 2$ is also a distance-$d$ independent set of $\Gp$ and vice versa.
	From the construction, for every $uv \in E(G)$, it follows that $\dist_{\Gp}(u, v) \leq d-1$ and therefore both $u$ and $v$ cannot be in the same distance-$d$ independent set of $\Gp$.
	Additionally, observe that for every $uv \notin E(G)$, any path between $u$ and $v$ in $\Gp$ must contain $x_p^{uw}$ and $x_p^{vz}$, for some $w \in N_G(u)$ and $z \in N_G(v)$.
	Thus, a shortest path between $u$ and $v$ in $\Gp$ must be of the form $x_0^{uw}\dots x_p^{uw}x_p^{vz}\dots x_0^{vz}$ if $d$ is odd and $x_0^{uw}\dots x_p^{uw}x^{\star}x_p^{vz}\dots x_0^{vz}$ if $d$ is even.
	One can verify that such a shortest path has length exactly $d$.
	Therefore, $\dist_{\Gp}(u, v) \geq d$.
	As a result, if $I$ is an independent set of $G$ (of size $k \geq 2$) then it is also a distance-$d$ independent set of $\Gp$.
	On the other hand, suppose that $\Ip$ is a distance-$d$ independent set of $\Gp$ of size $k \geq 2$.
	We claim that $\Ip$ does not contain any new vertex.
	Let $X = \bigcup_{uv \in E(G)}\{x_1^{uv}, \dots, x_{d-2}^{uv}\} \cup \{x^\star\}$ be the set of all new vertices.
	(Note that $x^\star$ only appears when $d$ is even.)
	To show that $\Ip \cap X = \emptyset$, we will show that the distance between a vertex in $X$ and any other vertex in $V(\Gp)$ is at most $d-1$.
	From the construction of $\Gp$, the distance between $x^\star$ (if exists) and any other vertex is at most $p + 1 = (d - 2)/2 + 1 \leq d-1$.
	It remains to show that one can always find a path of length at most $d-1$ between a new vertex $x = x_i^{uv} \in X - x^\star$ and a vertex $y = x_j^{wz} \in V(\Gp) - \{x, x^\star\}$, where $1 \leq i \leq d-2$, $0 \leq j \leq d-1$, and $uv, wz \in E(G)$.
	If $d$ is odd, such a path can be constructed by joining the paths from $x = x_i^{uv}$ to $x_p^{uv}$ (of length $\leq p-1$), from $x_p^{uv}$ to $x_p^{wz}$ (of length $\leq 1$), and from $x_p^{wz}$ to $y = x_j^{wz}$ (of length $\leq p$).
	If $d$ is even, such a path can be constructed by joining the paths from $x = x_i^{uv}$ to the vertex $a \in \{x_p^{uv}, x_{p+1}^{uv}\}$ closest to $x$ (of length $\leq p-1$), from $a$ to $x^\star$ (of length $\leq 1$), from $x^\star$ to the vertex $b \in \{x_p^{wz}, x_{p+1}^{wz}\}$ closest to $y$ (of length $\leq 1$), and from $b$ to $y = x_j^{wz}$ (of length $\leq p$).
	As a result, $\Ip \subseteq V(G)$, and from the construction of $\Gp$, it follows that $\Ip$ is an independent set of $G$ of size $k \geq 2$.
	
	We are now ready to show that \RR{$(G, I, J, \sfTJ)$} is a yes-instance of \textsc{ISR} if and only if \RR{$(\Gp, I, J, \sfTJ)$} is a yes-instance of \textsc{D$d$ISR}.
	Suppose that $\vert I \vert = \vert J \vert = k$.
	If $k = 1$, the claim is trivial. 
	Therefore, we consider $k \geq 2$.
	In this case, since any independent set of $G$ is also a distance-$d$ independent set of $\Gp$ and vice versa, it follows that any $\sfTJ$-sequence in $G$ is also a $\sfTJ$-sequence in $\Gp$ and vice versa.
	Our proof is complete.
\end{proof}
\end{toappendix}

\section{Extending Some Known Results for $d = 2$}%
\label{sec:extend-results}

In this section, we prove that several known results on the complexity of \textsc{D$d$ISR} for the case $d = 2$ can be extended for $d \geq 3$.

\begin{toappendix}
\subsection{General Graphs}%
\label{sec:general}

Ito~et~al.~\cite{ItoDHPSUU11} proved that \textsc{ISR} is $\ttPSPACE$-complete on general graphs under $\sfTJ/\sfTAR$.
Indeed, their proof uses only maximum independent sets, which implies that any token-jump is also a token-slide~\cite{BonsmaKW14}, and therefore the $\ttPSPACE$-completeness also holds under $\sfTS$.
We will show that the reduction of Ito~et~al.~\cite{ItoDHPSUU11} can be extended for showing the $\ttPSPACE$-completeness of \textsc{D$d$ISR} for $d \geq 3$.

\begin{theorem}\label{thm:general-TSTJ}
	\textsc{D$d$ISR} is $\ttPSPACE$-complete under $\sfR \in \{\sfTS, \sfTJ\}$ for any $d \geq 3$.
\end{theorem}
\begin{proof}
	Recall that Ito~et~al.~\cite{ItoDHPSUU11} proved the $\ttPSPACE$-hardness of the problem for $d = 2$ by reducing from a known $\ttPSPACE$-complete problem called \textsc{3-Satisfiability Reconfiguration (3SAT-R)}~\cite{GopalanKMP09}.
	A \textit{3SAT formula} $\varphi$ consists of $n$ variables $x_1, \dots, x_n$ and $m$ clauses $c_1, \dots, c_m$, each clause contains at most three literals (a literal is either $x_i$ or $\overline{x_i}$).
	A \textit{truth assignment} $\phi$ is a way of assigning either $1$ (true) or $0$ (false) to each $x_i$.
	A truth assignment $\phi$ \textit{satisfies a clause} if at least one literal is $1$ and it \textit{satisfies the formula $\varphi$} if all clauses are satisfied.
	In a \textsc{3SAT-R} instance $(\varphi, \phi_s, \phi_t)$, two satisfying assignments $\phi$ and $\phi^\prime$ are adjacent if they differ in assigning exactly one variable.
	The \textsc{3SAT-R} problem asks whether there is an adjacent satisfying assignments for $\varphi$ between two given satisfying assignments $\phi_s$ and $\phi_t$.
	We will slightly modify their reduction to prove the $\ttPSPACE$-hardness for $d \geq 3$.
	
	We first describe the reduction by Ito~et~al.~\cite{ItoDHPSUU11}.
	Let $\varphi$ be a given 3SAT formula. 
	We construct a graph $G$ as follows.
	For each variable $x_i$ ($1 \leq i \leq n$) in $\varphi$, we add an edge $e_{x_i}$ to the graph; its two endpoints are labeled $x_i$ and $\overline{x_i}$. 
	Then, for each clause $c_j$ ($1 \leq j \leq m$) in $\varphi$, we add a clique whose nodes correspond to literals in $c_j$. 
	Finally, we add an edge between two nodes in different components if and only if the nodes correspond to opposite literals.
	Ito~et~al.~\cite{ItoDHPSUU11} proved that there is a one-to-one correspondence between satisfying assignments for $\varphi$ and maximum independent sets of size $n+m$ in $G$: $n$ vertices are chosen from the endpoints of edges corresponding to the variables; a literal is $1$ (true) if the corresponding endpoint is chosen.
	They uses this observation to show that there is a sequence of adjacent satisfying assignments between $\phi_s$ and $\phi_t$ if and only if there is a $\sfTJ$-sequence between two corresponding size-$(n+m)$ independent sets of $G$.
	(Indeed, they proved under $\sfTAR$ rule.
	However, since only independent sets of size at least $n+m-1$ are used, any $\sfTAR$-sequence in $G$ can be converted into a $\sfTJ$-sequence~\cite{KaminskiMM12}, the result under $\sfTJ$ also holds.)
	
	\begin{figure}[!ht]
		\centering
		\begin{adjustbox}{max width=\textwidth}
			\begin{tikzpicture}[every node/.style={circle, draw, thick, fill=white, minimum size=3mm}]
				\node[label={[label distance=-0.2cm]above:$\overline{x_1}$}] (bx1) at (0,0) {};
				\node[label={[label distance=-0.2cm]above:$x_1$}] (x1) at (1,0) {};
				\node[label={[label distance=-0.2cm]above:$\overline{x_2}$}] (bx2) at (3,0) {};
				\node[label={[label distance=-0.2cm]above:$x_2$}] (x2) at (4,0) {};
				\node[label={[label distance=-0.2cm]above:$\overline{x_3}$}] (bx3) at (6,0) {};
				\node[label={[label distance=-0.2cm]above:$x_3$}] (x3) at (7,0) {};
				
				\node[label={[label distance=-0.2cm]below:$x_1$}] (c1x1) at (-0.5,-3) {};
				\node[label={[label distance=-0.25cm]below:$\overline{x_2}$}] (c1bx2) at (0.5,-3) {};
				\node[label={[label distance=-0.25cm]below:$\overline{x_1}$}] (c2bx1) at (3,-3) {};
				\node[label={[label distance=-0.2cm]below:$x_3$}] (c2x3) at (4,-3) {};
				\node[label={[label distance=-0.2cm]right:$x_2$}] (c2x2) at (3.5,-2) {};
				\node[label={[label distance=-0.25cm]below:$\overline{x_2}$}] (c3bx2) at (6.5,-3) {};
				\node[label={[label distance=-0.25cm]below:$\overline{x_3}$}] (c3bx3) at (7.5,-3) {};
				
				\draw[thick] (x1) -- (bx1) (x2) -- (bx2) (x3) -- (bx3) (c1x1) -- (c1bx2) (c2bx1) -- (c2x2) -- (c2x3) -- (c2bx1) (c3bx2) -- (c3bx3);
				
				\path[draw=none] (c1x1) to[-] coordinate[pos=0.33] (c1x1-1) coordinate[pos=0.66] (c1x1-2) (bx1) (c1bx2) to[-] coordinate[pos=0.33] (c1bx2-1) coordinate[pos=0.66] (c1bx2-2) (x2) (c2bx1) to[-] coordinate[pos=0.33] (c2bx1-1) coordinate[pos=0.66] (c2bx1-2) (x1) (c2x2) to[-] coordinate[pos=0.2] (c2x2-1) coordinate[pos=0.8] (c2x2-2) (bx2) (c2x3) to[-] coordinate[pos=0.33] (c2x3-1) coordinate[pos=0.66] (c2x3-2) (bx3) (c3bx2) to[-] coordinate[pos=0.33] (c3bx2-1) coordinate[pos=0.66] (c3bx2-2) (x2) (c3bx3) to[-] coordinate[pos=0.33] (c3bx3-1) coordinate[pos=0.66] (c3bx3-2) (x3) (c1x1) to[bend right=60] coordinate[pos=0.33] (c1x1-3) coordinate[pos=0.66] (c1x1-4) (c2bx1) (c2x3) to[bend right=60] coordinate[pos=0.33] (c2x3-3) coordinate[pos=0.66] (c2x3-4) (c3bx3) ;

				\draw[very thick, dotted] (c1x1) -- (bx1) (c1bx2) -- (x2) (c2bx1) -- (x1) (c2x2) -- (bx2) (c2x3) -- (bx3) (c3bx2) -- (x2) (c3bx3) -- (x3) (c1x1) edge[bend right=45] (c2bx1) (c2x3) edge[bend right=45] (c3bx3);
				
			\end{tikzpicture}
		\end{adjustbox}
		\caption{An example of constructing the graphs $G$ and $\Gp$ from a 3SAT formula $\varphi$ used in~\cite{ItoDHPSUU11} having three variables $x_1$, $x_2$, and $x_3$ and three clauses $c_1 = \{x_1, \overline{x_2}\}$, $c_2 = \{\overline{x_1}, x_2, x_3\}$, and $c_3 = \{\overline{x_2}, \overline{x_3}\}$. Paths of length $d-1$ are represented by dotted edges.}
		\label{fig:reduce-3SAT-DdISR-general}
	\end{figure}
	
	Our modification is simple. 
	Instead of joining two nodes in different components by an edge, we join them by a path of length $d-1$.
	Let $\Gp$ be the resulting graph.
	\figurename~\ref{fig:reduce-3SAT-DdISR-general} describes an example of constructing $\Gp$ from the same 3SAT formula used by Ito~et~al.~\cite{ItoDHPSUU11}.
	Indeed, when $d = 2$, our constructed graph $\Gp$ and the graph $G$ constructed by Ito~et~al. are identical.
	One can verify that any independent set of $G$ of size $n+m$ is also a D$d$IS of $\Gp$.
	Additionally, note that no token can ever leave its corresponding component.
	To see this, choose a token $t$ and assume inductively that the statement holds for all other tokens, then $t$ can never jump/slide to any vertex outside its corresponding component because some other token must be moved first in order to make room for $t$.
	Thus, any token-jump is also a token-slide.
	(Recall that each component is either a single vertex, an edge, or a triangle.)
	Hence, \RR{$(G, I, J, \sfR)$} is a yes-instance of \textsc{ISR} if and only if \RR{$(\Gp, I, J, \sfR)$} is a yes-instance of \textsc{D$d$ISR}, where $I$ and $J$ are size-$(n+m)$ independent sets of $G$ (which are also D$d$ISs of $\Gp$) and $\sfR \in \{\sfTS, \sfTJ\}$.
	Our proof is complete.
\end{proof}
\end{toappendix}

\subsection{Perfect Graphs}%
\label{sec:perfect}

In this section, we prove the $\ttPSPACE$-completeness of \textsc{D$d$ISR} ($d \geq 3$) on perfect graphs by extending the corresponding known result for \textsc{ISR} of Kami{\'n}ski~et~al.~\cite{KaminskiMM12}.
\begin{theorem}\label{thm:perfect-TSTJ}
	\textsc{D$d$ISR} is $\ttPSPACE$-complete on perfect graphs under $\sfR \in \{\sfTS, \sfTJ\}$ for any $d \geq 3$.
\end{theorem}
\begin{proof}
	Recall that Kami{\'n}ski~et~al.~\cite{KaminskiMM12} proved the $\ttPSPACE$-hardness of the problem for $d = 2$ by reducing from a known $\ttPSPACE$-complete problem called \textsc{Shortest Path Reconfiguration (SPR)}~\cite{Bonsma13}.
	In a \textsc{SPR}'s instance $(G, u, v, P, Q)$, two shortest $uv$-paths are adjacent if one can be obtained from the other by exchanging exactly one vertex.
	The \textsc{SPR} problem asks whether there is a sequence of adjacent shortest $uv$-paths between $P$ and $Q$.
	We will slightly modify their reduction to prove the $\ttPSPACE$-hardness for $d \geq 3$.
	
	We first describe the reduction by Kami{\'n}ski~et~al.~\cite{KaminskiMM12}.
	Let $G$ be a given graph and let $u, v \in V(G)$.
	We delete all vertices and edges that does not appear in any shortest $uv$-path and let the resulting graph be $\Gtil$.
	Let $k = \dist_G(u, v)$ and for each $i \in \{0, \dots, k\}$ let $D_i$ be the set of vertices in $\Gtil$ of distance $i$ from $u$ and $k-i$ from $v$.
	Let $\Gp$ be the graph obtained from $\Gtil$ by turning each $D_i$ into a clique and complementing (i.e., if there is an edge between two vertices, remove it; otherwise, create one) the edges of $\Gtil$ between two consecutive layers $D_i$ and $D_{i+1}$.
	Kami{\'n}ski~et~al.~\cite{KaminskiMM12} proved that $G^\prime$ is a perfect graph.
	They also claimed that there is a one-to-one correspondence between $uv$-shortest paths of $G$ and independent sets of size $k+1$ of $\Gp$: the $k+1$ vertices of a shortest $uv$-path in $G$ become independent in $\Gp$.
	This is used for showing that there is a sequence of adjacent $uv$-shortest paths in $G$ if and only if there is a $\sfR$-sequence between size-$(k+1)$ independent sets of $G^\prime$, for $\sfR \in \{\sfTS, \sfTJ, \sfTAR\}$.
	
	Our modification is simple.
	We construct a graph $\Gpp$ from $\Gp$ as follows.
	First, we replace each edge between two consecutive layers $D_i$ and $D_{i+1}$ by a new path of length $d-1$, and let the resulting graph be $\widetilde{\Gp}$.
	For each $j \in \{1, \dots, d-2\}$, let $D_i^j$ be the set of vertices in $\widetilde{\Gp}$ of distance $j$ from some vertex in $D_i$ and distance $d - 1 - j$ from some vertex in $D_{i+1}$.
	Let $\Gpp$ be the graph obtained from $\widetilde{\Gp}$ by turning each $D_i^j$ into a clique.
	Since $\Gp$ is perfect, the above construction implies that $\Gpp$ is perfect too.
	Clearly the construction of both $\Gp$ and $\Gpp$ can be done in polynomial time.
	In fact, if $d = 2$, we have $\Gpp = \Gp$.
	(For example, see \figurename~\ref{fig:reduce-SPR-DdISR-perfect}.)
	
	\begin{figure}[!ht]
		\centering
		\begin{adjustbox}{max width=\textwidth}
			\begin{tikzpicture}[every node/.style={circle, draw, thick, fill=white, minimum size=3mm}]
				\begin{scope}
					\foreach \i/\x/\y in {0/0/0,1/1/0,2/2/0,3/3/0,4/4/0,5/5/0,11/1/1,21/2/1,31/3/1,41/4/1,32/3/-1,42/4/-1} {
						\ifthenelse{\i=0}{
							\node[label=below:$u$] (\i) at (\x, \y) {};
						}{
							\ifthenelse{\i=5}{
								\node[label=below:$v$] (\i) at (\x, \y) {};
							}{
								\node (\i) at (\x, \y) {};
							}
						}
					}
					
					\draw[thick] (0) -- (1) -- (2) -- (3) -- (4) -- (5) (0) -- (11) -- (21) -- (31) -- (41) -- (5) (2) -- (32) -- (42) -- (5) (32) -- (4) (1) -- (21) (2) -- (31) (31) -- (42) (31) -- (4);
					
					\node[rectangle, draw=none, fill=none] at (2.5,-1.7) {$\Gtil$};
				\end{scope}
				\begin{scope}[shift={(6,0)}]
					\foreach \i/\x/\y in {0/0/0,1/1/0,2/2/0,3/3/0,4/4/0,5/5/0,11/1/1,21/2/1,31/3/1,41/4/1,32/3/-1,42/4/-1} {
						\ifthenelse{\i=0}{
							\node[label=below:$u$] (\i) at (\x, \y) {};
						}{
							\ifthenelse{\i=5}{
								\node[label=below:$v$] (\i) at (\x, \y) {};
							}{
								\node (\i) at (\x, \y) {};
							}
						}
					}
					
					\draw[thick] (11) -- (2) (21) -- (3) (21) -- (32) (3) -- (41) (3) -- (42) (32) -- (41);
					
					\node[rectangle, draw=none, fill=none] at (2.5,-1.7) {$\Gp$ ($d = 2$)};
					
					\begin{scope}[on background layer]
						\foreach \i in {0,...,5} {
							\node[rectangle, fill=gray!30!white, rounded corners, minimum width=5mm, minimum height=2.5cm, label={[label distance=-0.2cm]above:$D_{\i}$}] at (\i) {};
						}
					\end{scope}
				\end{scope}
				
				\begin{scope}[shift={(0,-4)}]
					\foreach \i/\x/\y in {0/0/0,1/1/0,2/2/0,3/3/0,4/4/0,5/5/0,11/1/1,21/2/1,31/3/1,41/4/1,32/3/-1,42/4/-1} {
						\ifthenelse{\i=0}{
							\node[label=below:$u$] (\i) at (2*\x, \y) {};
						}{
							\ifthenelse{\i=5}{
								\node[label=below:$v$] (\i) at (2*\x, \y) {};
							}{
								\node (\i) at (2*\x, \y) {};
							}
						}
					}
					
					\draw[thick] (11) to[-] coordinate[pos=0.5] (11-2) (2) (21) to[-] coordinate[pos=0.5] (21-3) (3) (21) to[-] coordinate[pos=0.5] (21-32) (32) (3) to[-] coordinate[pos=0.5] (3-41) (41) (3) to[-] coordinate[pos=0.5] (3-42) (42) (32) to[-] coordinate[pos=0.5] (32-41) (41);
					
					\foreach \i in {11-2, 21-3, 21-32, 3-41, 3-42, 32-41} {
						\node[fill=gray] at (\i) {};
					}
					
					\node[rectangle, draw=none, fill=none] at (2*2.5,-1.7) {$\Gpp$ ($d = 3$)};
					
					\begin{scope}[on background layer]
						\foreach \i in {0,...,5} {
							\node[rectangle, fill=gray!30!white, rounded corners, minimum width=5mm, minimum height=2.5cm, label={[label distance=-0.2cm]above:$D_{\i}$}] at (\i) {};
						}
						\path[draw=none] (1) to[-] coordinate[pos=0.5] (12) (2) to[-] coordinate[pos=0.5] (23) (3) to[-] coordinate[pos=0.5] (34) (4);
						\foreach \i/\j/\k in {1/1/12,2/1/23,3/1/34} {
							\node[rectangle, fill=gray!30!white, rounded corners, minimum width=5mm, minimum height=2cm, label={[label distance=-0.2cm]above:$D_{\i}^{\j}$}] at (\k) {};
						}
					\end{scope}
				\end{scope}
			\end{tikzpicture}
		\end{adjustbox}
		\caption{An example of constructing the perfect graphs $\Gp$ and $\Gpp$ from $\Gtil$. Vertices in a light-gray box form a clique. Vertices in $V(\Gpp) - V(\Gp)$ are marked with the gray color.}
		\label{fig:reduce-SPR-DdISR-perfect}
	\end{figure}
	
	It is sufficient to show that there is a $\sfR$-sequence of size-$(k+1)$ independent sets in $\Gp$ if and only if there is a $\sfR$-sequence of size-$(k+1)$ distance-$d$ independent sets in $\Gpp$, for $\sfR \in \{\sfTS, \sfTJ\}$.
	From the construction of $\Gpp$, it follows that any independent set of $\Gp$ of size $k+1$ is also a distance-$d$ independent set of $\Gpp$, since a shortest path of length $2$ in $\Gp$ between two vertices $x \in D_i$ and $y \in D_{i+1}$ becomes a shortest path of length $(d-1) + 1 = d$ in $\Gpp$.
	Thus, any $\sfTJ$-sequence in $\Gp$ between two size-$(k+1)$ independent sets is also a $\sfTJ$-sequence in $\Gpp$.
	On the other hand, note that a size-$(k+1)$ distance-$d$ independent set of $\Gpp$ must not contain any new vertex in $V(\Gpp) - V(\Gp)$, otherwise one can verify that the set must be of size at most $k$, which is a contradiction.
	(Note that a D$d$IS must contain at most one vertex in $\bigcup_{j=1}^{d-2}D_i^j$.)
	As a result, a size-$(k+1)$ distance-$d$ independent set of $\Gpp$ is also an independent set of $\Gp$.
	Thus, any $\sfTJ$-sequence in $\Gpp$ between two size-$(k+1)$ distance-$d$ independent sets is also a $\sfTJ$-sequence in $\Gpp$.
	We note that in both $\Gp$ and $\Gpp$, a token can never move out of the layer where it belongs, and therefore any token-jump is also a token-slide.
	Our proof is complete.
\end{proof}

\subsection{Planar Graphs}%
\label{sec:planar}

In this section, we claim that the $\ttPSPACE$-hard reduction of Hearn and Demaine~\cite{HearnD05} for \textsc{ISR} under $\sfTS$ can be extended to \textsc{D$d$ISR} ($d \geq 3$) under $\sfR \in \{\sfTS, \sfTJ\}$.
We now briefly introduce the powerful tool Hearn and Demaine~\cite{HearnD05,HearnD09} used as the source problem for proving several hardness results, including the $\ttPSPACE$-hardness of \textsc{ISR}: the \textit{Nondeterministic Constraint Logic (NCL)} machine.
\begin{figure}[!ht]
	\centering
	\begin{adjustbox}{max width=\textwidth}
		\begin{tikzpicture}
			\begin{scope}
				\node (v1) at (1,3) [circle, draw] {};
				\node (v2) at (5,3) [circle, draw] {};
				\node (v3) at (2,2) [circle, draw] {};
				\node (v4) at (4,2) [circle, draw] {};
				\node (v5) at (1,1) [circle, draw] {};
				\node (v6) at (5,1) [circle, draw] {};
				\draw (v1) -- node[midway, above left]{2} (v2) [ultra thick] {};
				\draw (v1) -- node[midway, above right]{2} (v3) [ultra thick] {};
				\draw (v1) -- node[midway, above left]{2} (v5) [ultra thick] {};
				\draw (v3) -- node[midway, above left]{2} (v4) [ultra thick] {};
				\draw (v3) -- node[midway, right]{2} (v5) [ultra thick] {};
				\draw (v5) -- node[midway, above left]{2} (v6) [ultra thick] {};
				\draw (v2) -- node[midway, left]{1} (v4)  {};
				\draw (v2) -- node[midway, above right]{1} (v6)  {};
				\draw (v4) -- node[midway, left]{1} (v6)  {};
				\draw (3.1,3.1) -- (2.9,3.0) -- (3.1,2.9) [ultra thick] {};
				\draw (2.9,2.1) -- (3.1,2.0) -- (2.9,1.9) [ultra thick] {};
				\draw (2.9,1.1) -- (3.1,1.0) -- (2.9,0.9) [ultra thick] {};
				\draw (0.9,2.1) -- (1.0,1.9) -- (1.1,2.1) [ultra thick] {};
				\draw (4.9,1.9) -- (5.0,2.1) -- (5.1,1.9) {};
				\draw (1.58,2.58) -- (1.65,2.35) -- (1.42,2.42) [ultra thick] {};
				\draw (1.42,1.58) -- (1.35,1.35) -- (1.58,1.42) [ultra thick] {};
				\draw (4.42,2.58) -- (4.65,2.65) -- (4.58,2.42) {};
				\draw (4.42,1.42) -- (4.65,1.35) -- (4.58,1.58) {};
				
				\node[rectangle, fill=none, draw=none] at (2.5,0.5) {(a)};
			\end{scope}
			\begin{scope}[shift={(7,1)}]
				\node (v0) at (1.0,1.0) [circle, draw] {};
				\node (v1) at (0.2,0.5) [label=above:1] {};
				\node (v2) at (1.8,0.5) [label=above:1] {};
				\node (v3) at (1.0,2.0) [label=below right:2] {};
				\draw (v0) -- (v1) {};
				\draw (v0) -- (v2) {};
				\draw (v0) -- (v3) [ultra thick] {};
				
				\node[rectangle, fill=none, draw=none] at (1,0.3) {(b)};
			\end{scope}
			\begin{scope}[shift={(10,1)}]
				\node (v0) at (1.0,1.0) [circle, draw] {};
				\node (v1) at (0.2,0.5) [label=above:2] {};
				\node (v2) at (1.8,0.5) [label=above:2] {};
				\node (v3) at (1.0,2.0) [label=below right:2] {};
				\draw (v0) -- (v1) [ultra thick] {};
				\draw (v0) -- (v2) [ultra thick] {};
				\draw (v0) -- (v3) [ultra thick] {};
				
				\node[rectangle, fill=none, draw=none] at (1,0.3) {(c)};
			\end{scope}
		\end{tikzpicture}
	\end{adjustbox}
	\caption{(a) A configuration of an NCL machine, (b) NCL \textsc{And} vertex, and (c) NCL \textsc{Or} vertex. \RR{Edges of weight $1$ (resp., $2$) are thin (resp., thick).}}%
	\label{fig:NCL}
\end{figure}
An NCL ``machine'' is an undirected graph whose \RR{edges are assigned with weight $1$ (thin) or $2$ (thick)}. 
An (\emph{NCL}) \emph{configuration} of this machine is an orientation of the edges such that the sum of weights of in-coming arcs at each vertex is at least two. 
\figurename~\ref{fig:NCL}(a) illustrates a configuration of an NCL machine. 
Two NCL configurations are \emph{adjacent} if they differ in a single edge direction. 
Given an NCL machine along with two configurations, Hearn and Demaine~\cite{HearnD05} proved that it is $\ttPSPACE$-complete to determine whether there exists a sequence of adjacent NCL configurations which transforms one into the other.  

Indeed, even when restricted to \emph{\textsc{And}/\textsc{Or} constraint graphs}, the above problem remains $\ttPSPACE$-complete.
An \emph{\textsc{And}/\textsc{Or} constraint graph} is a NCL machine that contains exactly two type of vertices: ``NCL \textsc{And} vertices'' and ``NCL \textsc{Or} vertices''.

More precisely, a vertex of degree three is called an \emph{NCL \textsc{And} vertex} if its three incident edges have weights $1$, $1$ and $2$. 
(See \figurename~\ref{fig:NCL}(b).)
The ``behavior'' of an NCL \textsc{And} vertex $u$ is similar to a logical \textsc{and}: 
the weight-$2$ edge can be directed outward for $u$ if and only if both two weight-$1$ edges are directed inward for $u$. 
Note that, however, the weight-$2$ edge is not necessarily directed outward even when both weight-$1$ edges are directed inward. 
A vertex of degree three is called an \emph{NCL \textsc{Or} vertex} if its three incident edges have weights $2$, $2$ and $2$. 
(See \figurename~\ref{fig:NCL}(c).)
The ``behavior'' of an NCL \textsc{Or} vertex $v$ is similar to a logical \textsc{or}: 
one of the three edges can be directed outward for $v$ if and only if at least one of the other two edges is directed inward for $v$. 

For example, the NCL machine in \figurename~\ref{fig:NCL}(a) is an \textsc{And}/\textsc{Or} constraint graph. 
\begin{theorem}\label{thm:planar-TSTJ}
	\textsc{D$d$ISR} is $\ttPSPACE$-complete under $\sfR \in \{\sfTS, \sfTJ\}$ on planar graphs of maximum degree three and bounded bandwidth for any $d \geq 2$.
\end{theorem}
\begin{proof}
	We note that Hearn and Demaine~\cite{HearnD05}'s proof can be applied with the vertex gadgets shown in \figurename~\ref{fig:ncl-gadgets-DdISR}. 
	Indeed, when $d = 2$, the gadgets in \figurename~\ref{fig:ncl-gadgets-DdISR} are exactly those used in~\cite{HearnD05}.
	We rephrase their proof here for the sake of completeness.
	From now on we consider only $\sfTS$ rule. 
	We will see later that in the constructed graph, any token-jump is also a token-slide, and therefore our reduction also holds for $\sfTJ$.
	The maximum degree will also be clear from the construction of gadgets and how they are joined.
	The final complexity result follows from the result of van~der~Zanden~\cite{Zanden15} which combines~\cite{HearnD05} and~\cite{Wrochna18}: NCL remains $\ttPSPACE$-complete even if an input NCL machine is planar and bounded bandwidth.
	
	The NCL \textsc{And} and \textsc{Or} vertex gadgets are constructed as in \figurename~\ref{fig:ncl-gadgets-DdISR}(a) and (b). 
	The edges that cross the dashed-line gadget borders are ``port'' edges. 
	A token on an outer port-edge vertex represents an inward-directed NCL edge, and vice versa. 
	Given an \textsc{And/Or} graph and configuration, we construct a corresponding graph for \textsc{D$d$ISR} under $\sfTS$, by joining together \textsc{And} and \textsc{Or} vertex gadgets at their shared port edges, placing the port tokens appropriately.
	
	\begin{figure}[!ht]
		\centering
		\begin{adjustbox}{max width=\textwidth}
			\begin{tikzpicture}[scale=0.75, every node/.style={circle, draw, thick, fill=white, minimum size=5mm, transform shape}]
				\begin{scope}
					\foreach \i/\x/\y in {1/0/1,2/0/0,3/-2/-3,4/2/-3,5/-2.5/-4,6/2.5/-4,7/-2.5/-5,8/2.5/-5} {
						\node (\i) at (\x, \y) {};
					}
					
					\draw[thick] (1) -- (2) (5) -- (7) (6) -- (8) (3) -- (5) (4) -- (6);
					\draw[thick, dotted] (2) -- (3) (3) -- (4) (2) -- (4);

					\foreach \i in {1,5,6} {
						\node[fill=black, minimum size=2mm] at (\i) {};
					}
					
					\draw[very thick, dashed, rounded corners=30] (0,1) -- ([shift={(-1,0)}]-2.5,-4.5) -- ([shift={(1,0)}]2.5,-4.5) -- cycle ;
					
					\node[rectangle, draw=none, fill=none] at (0,-5.5) {(a)};
				\end{scope}
				\begin{scope}[shift={(7.5,0.5)}]
					\foreach \i/\x/\y in {1/0/3,2/0/2,3/0/-1,4/0/-2,5/-1.0/-3,6/1.0/-3,7/-1.0/-6,8/1.0/-6,9/-1.0/-7,10/1.0/-7,11/-1.0/-8,12/1.0/-8} {
						\node (\i) at (\x, \y) {};
					}

					\draw[thick] (1) -- (2) (3) -- (4) -- (5) -- (6) (4) -- (6) (7) -- (9) -- (11) (8) -- (10) -- (12);
					\draw[thick, dotted] (2) -- (3) (5) -- (7) (6) -- (8);

					\foreach \i in {1,4,9,10} {
						\node[fill=black, minimum size=2mm] at (\i) {};
					}
					
					\draw[very thick, dashed, rounded corners=30] (0,3.2) -- ([shift={(-0.7,0)}]-2.5,-7.5) -- ([shift={(0.7,0)}]2.5,-7.5) -- cycle;
					
					\node[rectangle, draw=none, fill=none] at (0,-8.5) {(b)};
				\end{scope}
			\end{tikzpicture}
		\end{adjustbox}
		\caption{Vertex gadgets for \textsc{D$d$ISR} under $\sfTS$: (a) \textsc{And}, (b) \textsc{Or}. \RR{The dotted edges represent paths of length $d-2$.}}
		\label{fig:ncl-gadgets-DdISR}
	\end{figure}
	
	First, observe that no port token may ever leave its port edge. 
	Choosing a particular port edge $E$, if we inductively assume that this condition holds for all other port edges, then there is never a legal move outside $E$ for its token---another port token would have to leave its own edge first.
	
	The \textsc{And} gadget clearly satisfies the same constraints as an NCL \textsc{And} vertex; the upper token can slide in just when both lower tokens are slid out. 
	Likewise, the upper token in the \textsc{Or} gadget can slide in when either lower token is slid out---the internal token can then slide to one side or the other to make room. 
	It thus satisfies the same constraints as an NCL \textsc{Or} vertex.
	As a result, it follows that a sequence of adjacent NCL configurations in a given \textsc{And/Or} graph can indeed be transformed into a $\sfTS$-sequence in the constructed graph and vice versa.
	Since no port token ever leaves its port edge, it follows from the construction of gadgets that in the constructed graph any token-jump is also a token-slide, which means the theorem also holds for $\sfTJ$. 
	Our proof is complete.
\end{proof}

\section{Open Problem: Trees}%
\label{sec:trees-TS}

In this section, we discuss the problems on trees.
Since the power of a tree is a (strongly) chordal graph~\cite{LinS95,KearneyC98} and \textsc{ISR} on chordal graphs under $\sfTJ$ is in $\ttP$~\cite{KaminskiMM12}, Proposition~\ref{prop:results-inherit-from-ISR-TJ} implies that \RR{the following corollary holds}.
\begin{corollary}\label{cor:trees-TJ}
	\textsc{D$d$ISR} under $\sfTJ$ on trees is in $\ttP$ for any $d \geq 3$.
\end{corollary}
On the other hand, the complexity of \textsc{D$d$ISR} under $\sfTS$ for $d \geq 3$ remains unknown.
\begin{conjecture}\label{conj:TS-tree}
	\textsc{D$d$ISR} under $\sfTS$ on trees is in $\ttP$ for $d \geq 3$.
\end{conjecture}
Demaine~et~al.~\cite{DemaineDFHIOOUY15} showed that the problem for $d = 2$ can be solved in linear time.
Their algorithm is based on the so-called \textit{rigid tokens}.
Given a tree $T$ and a D$d$IS $I$ of $T$ ($d \geq 2$), a token $t$ on $u \in I$ is \textit{$(T, I, d)$-rigid} if there is no $\sfTS$-sequence in $T$ that slides $t$ from $u$ to some vertex $v \in N_T(u)$.
We denote by $\calR(T, I, d)$ the set of all vertices where $(T, I, d)$-rigid tokens are placed.
Demaine~et~al. proved that $\calR(T, I, 2)$ can be found in linear time.
(In general, deciding whether a token is $(G, I, 2)$-rigid on some graph $G$ is $\ttPSPACE$-complete~\cite{Fox-EpsteinHOU15}.)
Clearly, it holds for any $d \geq 2$ that every \textsc{D$d$ISR}'s instance \RR{$(T, I, J, \sfTS)$} where $\calR(T, I, d) \neq \calR(T, J, d)$ is a no-instance.
When $\calR(T, I, 2) = \calR(T, J, 2) = \emptyset$, they proved that \RR{$(T, I, J, \sfTS)$} is always a yes-instance of \textsc{ISR}.
Based on these observations, one can derive a polynomial-time algorithm for solving \textsc{ISR} on trees under $\sfTS$.

Indeed, one can extend the recursive characterization of $(T, I, 2)$-rigid tokens given by Demaine~et~al. as follows.
For a given vertex $u$ of a graph $G$ and an integer $s \geq 0$, we denote by $N^s_G(u)$ the \textit{$s$-neighborhood of $u$ in $G$} which contains all vertices of distance \textit{exactly} $s$ from $u$ in $G$.
Additionally, we denote by $N^s_G[u] = \bigcup_{i=0}^{s}N^i_G(u)$ the set of all vertices of distance at most $s$ from $u$ in $G$ and called it the \textit{s-closed neighborhood of $u$ in $G$}.
For example, $N^0_G(u) = \{u\}$, $N^1_G(u) = N_G(u)$, and $N^1_G[u] = N_G[u]$.
For two vertices $u, v$ of a tree $T$, we denote by $T^u_v$ the tree induced by $v$ and its descendants when regarding $u$ as the root of $T$.
\begin{lemma}\label{lem:TId-rigid}
	Let $I$ be a D$d$IS of a tree $T$ for $d \geq 2$ and let $u \in I$.
	The token $t$ on $u$ is $(T, I, d)$-rigid if and only if
	\begin{enumerate}[(a)]
		\item $\vert V(T) \vert = \vert \{u\} \vert = 1$; or
		\item $\vert V(T) \vert \geq 2$ and for each $v \in N_T(u)$ there exists a vertex $w \in (N^{d-1}_T(v) - u) \cap I$ such that the token $t_w$ on $w$ is $(T^v_w, I \cap V(T^v_w), d)$-rigid. 
	\end{enumerate}
\end{lemma}
\begin{proof}
	Our proof is analogous to the one given by Demaine~et~al.~\cite[Lemma~1]{DemaineDFHIOOUY15}.
	For the sake of completeness, we present it here.
	Part (a) is trivial.
	We prove part (b).
	Before going into more details, we note that for any $w \in (N^{d-1}_T(v) - u)$, if $w \notin V(T^u_v)$ then $\dist_T(w, u) \leq d - 2$ and therefore $w \notin I$.
	In other words, $(N^{d-1}_T(v) - u) \cap I \subseteq V(T^u_v)$.
	\begin{itemize}
		\item[($\Leftarrow$)] Suppose to the contrary that there exists a $\sfTS$-sequence that slides $t$ to some vertex $v \in N_T(u)$.
		On the other hand, from the assumption, there exists a vertex $w \in (N^{d-1}_T(v) - u) \cap I$ such that the token $t_w$ on $w$ is $(T^v_w, I \cap V(T^v_w), d)$-rigid. 
		Since $\dist_T(u, w) = d$, $t_w$ must be moved before $t$ and therefore can only be moved to some vertex in $T^v_w$.
		However, since the token $t_w$ on $w$ is $(T^v_w, I \cap V(T^v_w), d)$-rigid, this cannot be done.
		The reason is that if there is a $\sfTS$-sequence $\calS$ in $T$ that slides $t_w$, we can ``restrict'' it to a $\sfTS$-sequence in $T^v_w$ by taking the intersection of each independent set in $\calS$ with $V(T^v_w)$, which implies $t_w$ is $(T^v_w, I \cap V(T^v_w), d)$-movable.
		
		\item[($\Rightarrow$)] Suppose that $|V(T)| \geq 2$.
			We claim that if there exists $v \in N_T(u)$ such that either $(N^{d-1}_T(v) - u) \cap I = \emptyset$ or for every $w \in (N^{d-1}_T(v) - u) \cap I$, the token $t_w$ on $w$ is $(T^v_w, I \cap V(T^v_w), d)$-movable then there exists a $\sfTS$-sequence $\calS$ in $T$ that moves the token $t$ on $u$ to $v$. 
			If $(N^{d-1}_T(v) - u) \cap I = \emptyset$, there are no tokens of distance $d - 1$ from $v$ except $t$, and therefore we can immediately slide $t$ to $v$.
			Thus, it remains to consider the case for every $w \in (N^{d-1}_T(v) - u) \cap I$, the token $t_w$ on $w$ is $(T^v_w, I \cap V(T^v_w), d)$-movable.
			Since $t_w$ is $(T^v_w, I \cap V(T^v_w), d)$-movable, there exists a $\sfTS$-sequence $\calS_w$ in $T^v_w$ that slides $t_w$ to some vertex in $N_{T^v_w}(w)$.
			Moreover, note that for any D$d$IS $I_w$ of $T^v_w$, we have $I_w \cup (I^\prime - V(T^v_w))$ is also a D$d$IS of $T$, for any D$d$IS $I^\prime$ of $G$ satisfying $I^\prime \cap V(T^v_w) = I \cap V(T^v_w)$.
			(Intuitively, we move tokens ``away from $w$'' inside $T^v_w$. 
			These moves do not break any ``distance-$d$ restriction'' between a token inside $T^v_w$ and a token outside $T^v_w$.)
			In other words, one can ``extend'' $\calS_w$ to form a $\sfTS$-sequence in $T$ by simply taking the union of each D$d$IS in $\calS_w$ with $I^\prime - V(T^v_w)$.
			Thus, for each $w \in (N^{d-1}_T(v) - u) \cap I$, we can perform $\calS_w$ to slides $t_w$ to some vertex in $N_{T^v_w}(w)$ and finally slide $t$ from $u$ to $v$.
			It remains to verify that for any pair of distinct vertices $w, w^\prime \in (N^{d-1}_T(v) - u) \cap I$ we can indeed perform $\calS_w$ and $\calS_{w^\prime}$ independently.
			To see this, note that by definition the subtrees $T^v_w$ and $T^{v}_{w^\prime}$ are distinct, i.e., $V(T^v_w) \cap V(T^v_{w^\prime}) = \emptyset$.
			Moreover, since both $w, w^\prime$ are in $I$, there distance must be at least $d$.
	\end{itemize}
\end{proof}
Using this characterization, one can naturally derive a recursive algorithm that decides if the token on $u \in I$ is $(T, I, d)$-rigid in $O(n)$ time, where $n = \vert V(T) \vert$.
Therefore, one can find $\calR(T, I, d)$ in $O(n^2)$ time.

However, for $d \geq 3$, even when $\calR(T, I, d) = \calR(T, J, d) = \emptyset$, \RR{$(T, I, J, \sfTS)$} may be a no-instance of \textsc{D$d$ISR}.
An example of such instances is described in \figurename~\ref{fig:exa-no-instance-trees-TS}.
\begin{figure}[ht]
	\centering
	\begin{adjustbox}{max width=\textwidth}
		\begin{tikzpicture}[every node/.style={circle, draw, thick, fill=white, minimum size=5mm}]
			\foreach \i/\x/\y in {1/0/0, 2/1/0, 3/2/0, 4/3/0, 5/5/0, 6/6/0, 7/7/0, 8/8/0, 9/0/-2, 10/1/-2, 11/2/-2, 12/3/-2, 13/5/-2, 14/6/-2, 15/7/-2, 16/8/-2, 17/4/-1} {
				\ifthenelse{\i=17}{
					\node[label=below:$u$] (\i) at (\x, \y) {};
				}{
					\node (\i) at (\x, \y) {};
				}
			}
			\foreach \i in {1, 9} {
				\node[minimum size=3mm, fill=black] at (\i) {}; 
			}
			\foreach \i in {8, 16} {
				\node[minimum size=3mm, fill=gray] at (\i) {}; 
			}
			\draw[thick] (1) -- (2) (3) -- (4) (5) -- (6) (7) -- (8) (9) -- (10) (11) -- (12) (13) -- (14) (15) -- (16) (17) -- (4) (17) -- (5) (17) -- (12) (17) -- (13);
			\draw[thick, dotted] (2) -- (3) (6) -- (7) (10) -- (11) (14) -- (15);
		\end{tikzpicture}
	\end{adjustbox}
	\caption{A no-instance \RR{$(T, I, J, \sfTS)$} of \textsc{D$d$ISR} ($d \geq 3$) with $\calR(T, I, d) = \calR(T, J, d) = \emptyset$. Tokens in $I$ (resp. $J$) are marked with the black (resp. gray) color. All tokens are of distance $d-1$ from $u$.}
	\label{fig:exa-no-instance-trees-TS}
\end{figure}
We remark that the instance described in \figurename~\ref{fig:exa-no-instance-trees-TS} does not apply for $d = 2$, because in that case no tokens can move.
As a result, characterizing these no-instances may be crucial in proving Conjecture~\ref{conj:TS-tree}.
\RA{Additionally, a first step toward proving Conjecture~\ref{conj:TS-tree} may be to consider \textsc{D$d$ISR} under $\sfTS$ on \textit{caterpillars}---a subclass of trees. Here a \textit{caterpillar} is a tree having a \textit{backbone path} containing all vertices of degree at least $3$ (called \textit{backbone vertices}).
Each \textit{hair} of the caterpillar is a path starting from a backbone vertex and not containing any other backbone vertex. 
For example, the graph in \figurename~\ref{fig:exa-no-instance-trees-TS} is a caterpillar with the backbone path having exactly one vertex $u$ (whose degree is at least $3$) and four hairs of length exactly $d-1$ starting from $u$.}

\section*{Acknowledgements}

\RA{We thank the anonymous reviewers for their valuable comments and suggestions which help to improve earlier versions of this paper.}
This research is partially supported by the Japan Society for the Promotion of Science (JSPS) KAKENHI Grant Number JP20H05964 (AFSA).

\printbibliography

\end{document}